\documentclass{sig-alternate}
\usepackage{mathtools}

\usepackage{url}
\usepackage{multirow}
\usepackage{epsfig}
\usepackage{epstopdf}
\usepackage{graphicx}
\usepackage{amssymb}
\usepackage{amsmath}
\usepackage{verbatim}
\usepackage{algorithm}
\usepackage[noend]{algorithmic}
\usepackage{dsfont}
\usepackage{tipa}
\usepackage{balance}
\usepackage{caption}
\usepackage{paralist} 
\DeclareCaptionType{copyrightbox}
\usepackage{subfig}
\usepackage[show]{notes}

\newcommand{\spara}[1]{\smallskip\noindent{\bf{#1}}}

\newcommand{\calL}{{\ensuremath{\cal L}}}
\newcommand{\graph}{{\ensuremath{G}}}
\newcommand{\vertices}{{\ensuremath{V}}}
\newcommand{\novertices}{{\ensuremath{n}}}
\newcommand{\vertex}{{\ensuremath{v}}}
\newcommand{\edges}{{\ensuremath{E}}}
\newcommand{\aset}{{\ensuremath{A}}}
\newcommand{\noedges}{{\ensuremath{m}}}
\newcommand{\source}{{\ensuremath{s}}}
\newcommand{\dest}{{\ensuremath{t}}}
\newcommand{\traffic}{{\ensuremath{w}}}
\newcommand{\nopairs}{{\ensuremath{k}}}
\newcommand{\backbone}{{\ensuremath{R}}}
\newcommand{\dist}{{\ensuremath{d}}}

\newcommand{\sfactor}{{\ensuremath{\lambda}}}
\newcommand{\ccost}{{\ensuremath{c}}}
\newcommand{\ppath}{{\ensuremath{p}}}
\newcommand{\edge}{{\ensuremath{e}}}
\newcommand{\budget}{{\ensuremath{B}}}

\newcommand{\ccset}{{\ensuremath{{\cal C}}}}
\newcommand{\cc}{{\ensuremath{C}}}
\newcommand{\feasible}{{\ensuremath{F}}}

\newcommand{\cut}{{\ensuremath{S}}}
\newcommand{\cutset}{{\ensuremath{{\cal S}}}}
\newcommand{\vertexpairs}{{\ensuremath{\vertices_2}}}
\newcommand{\EBC}{{\ensuremath{EB}}}
\newcommand{\EBClog}{{\ensuremath{EB_{\calL}}}}

\newcommand{\nsp}{{\ensuremath{\sigma}}}

\newcommand{\bigO}{{\ensuremath{\cal O}}}
\newcommand{\NP}{{\ensuremath{\mathbf{NP}}}}
\newcommand{\NPhard}{{\NP-hard}}
\newcommand{\NPcomp}{{\NP-complete}}
\newcommand{\benefit}{{\ensuremath{b}}}
\newcommand{\effective}{{\ensuremath{\widehat{\ell}}}}

\newcommand{\HM}{{\ensuremath{H}}}
\newcommand{\universe}{{\ensuremath{U}}}
\newcommand{\e}{{\ensuremath{u}}}
\newcommand{\noe}{{\ensuremath{n}}}
\newcommand{\collection}{{\ensuremath{\cal S}}}
\newcommand{\set}{{\ensuremath{S}}}
\newcommand{\nos}{{\ensuremath{m}}}
\newcommand{\bk}{{\ensuremath{k}}}
\newcommand{\setnode}{{\ensuremath{v}}}
\newcommand{\terminal}{{\ensuremath{z}}}
\newcommand{\iter}{{\ensuremath{I}}}

\newcommand{\landmark}{{\ensuremath{z}}}
\newcommand{\landmarks}{{\ensuremath{L}}}
\newcommand{\nolandmarks}{{\ensuremath{\ell}}}

\newcommand{\steinerforest}{{\sc Steiner\-Forest}}
\newcommand{\kspanner}{{\ensuremath{k}-{\sc Spanner}}}
\newcommand{\setcover}{{\sc Set\-Cover}}
\newcommand{\backbonediscovery}{{\sc Back\-bone\-Discovery}}

\newcommand{\greedy}{{\sf Greedy}}
\newcommand{\greedyeb}{{\sf Greedy\-EB}}

\newcommand{\london}{{\sf London\-Tube}}
\newcommand{\flights}{{\sf US\-Flights}}
\newcommand{\ukroad}{{\sf UK\-Road}}
\newcommand{\sjroad}{{\sf SJ\-Road}}
\newcommand{\nyctaxi}{{\sf NYC\-Taxi}}
\newcommand{\wikispeedia}{{\sf Wikispeedia}}
\newcommand{\abeline}{{\sf Abeline}}

\newtheorem{problem}{Problem}
\newtheorem{theorem}{Theorem}
\newtheorem{lemma}{Lemma}

\newcommand{\squishlist}{\begin{list}{$\bullet$}
  { \setlength{\itemsep}{0pt}
     \setlength{\parsep}{3pt}
     \setlength{\topsep}{3pt}
     \setlength{\partopsep}{0pt}
     \setlength{\leftmargin}{1.5em}
     \setlength{\labelwidth}{1em}
     \setlength{\labelsep}{0.5em} } }
\newcommand{\squishend}{
  \end{list}  }

\begin{document}


\title{Discovering the Network Backbone from Traffic Activity Data}

\author{
%
%
\alignauthor Sanjay Chawla \\
\affaddr University of Sydney \\
\affaddr Sydney, Australia \\
\email sanjay.chawla@sydney.edu.au \\
\and
\alignauthor Venkata Rama Kiran Garimella \\ 
\affaddr Aalto University \\
\affaddr Helsinki, Finland \\
\email kiran.garimella@aalto.fi
\and
\alignauthor Aristides Gionis \\
\affaddr Aalto University and HIIT \\
\affaddr Helsinki, Finland \\
\email aristides.gionis@aalto.fi \\
\and
\alignauthor Dominic Tsang \\
\affaddr  University of Sydney  \\
\affaddr Sydney, Australia \\
\email dwktsang@yahoo.com
}

\maketitle

\begin{abstract}
We introduce a new computational problem, the \backbonediscovery\ problem,
which encapsulates both {\em functional} and {\em structural} aspects
of network analysis. 
While the topology of a typical road network
has been available for a long time (e.g., through maps),
it is only recently that fine-granularity functional (activity and usage) 
information 
about the network (like source-destination traffic information)
is being collected and is readily available.
The combination of functional and structural information provides
an efficient way to explore and
understand usage patterns of networks and aid in design and
decision making. We propose efficient algorithms for
the \backbonediscovery\ problem including a novel use
of edge centrality.
We observe that for many real world networks, our algorithm produces a backbone with a small subset of the edges that support a large percentage of the network activity.


\end{abstract}

\section{Introduction}

In this paper we study a novel problem, 
which combines {\em structural} and {\em functional (activity)} network data.
In recent years there has been a large body of research related
to exploiting structural information of networks. However, with the increasing
availability of fine-grained functional information, it is now
possible to obtain a detailed understanding of activities on
a network. 
Such activities include source-destination traffic
information in road and communication networks.

More specifically we study the problem of discovering the {\em backbone} of
traffic networks. 
In our setting, we consider the topology of a network
$\graph=(\vertices,\edges)$ and a traffic log 
$\calL=\{(\source_i, \dest_i, \traffic_i)\}$, 
recording the amount of traffic $\traffic_i$ that incurs between
source $\source_i$ and destination~$\dest_i$. 
We are also given a budget \budget\ that accounts 
for a total edge cost. 
The goal is to discover a sparse subnetwork \backbone\ of \graph, 
of cost at most \budget, which summarizes as well as possible the
recorded traffic~\calL.

\begin{figure*}[t]
\begin{center}
\includegraphics[width=0.40\textwidth]{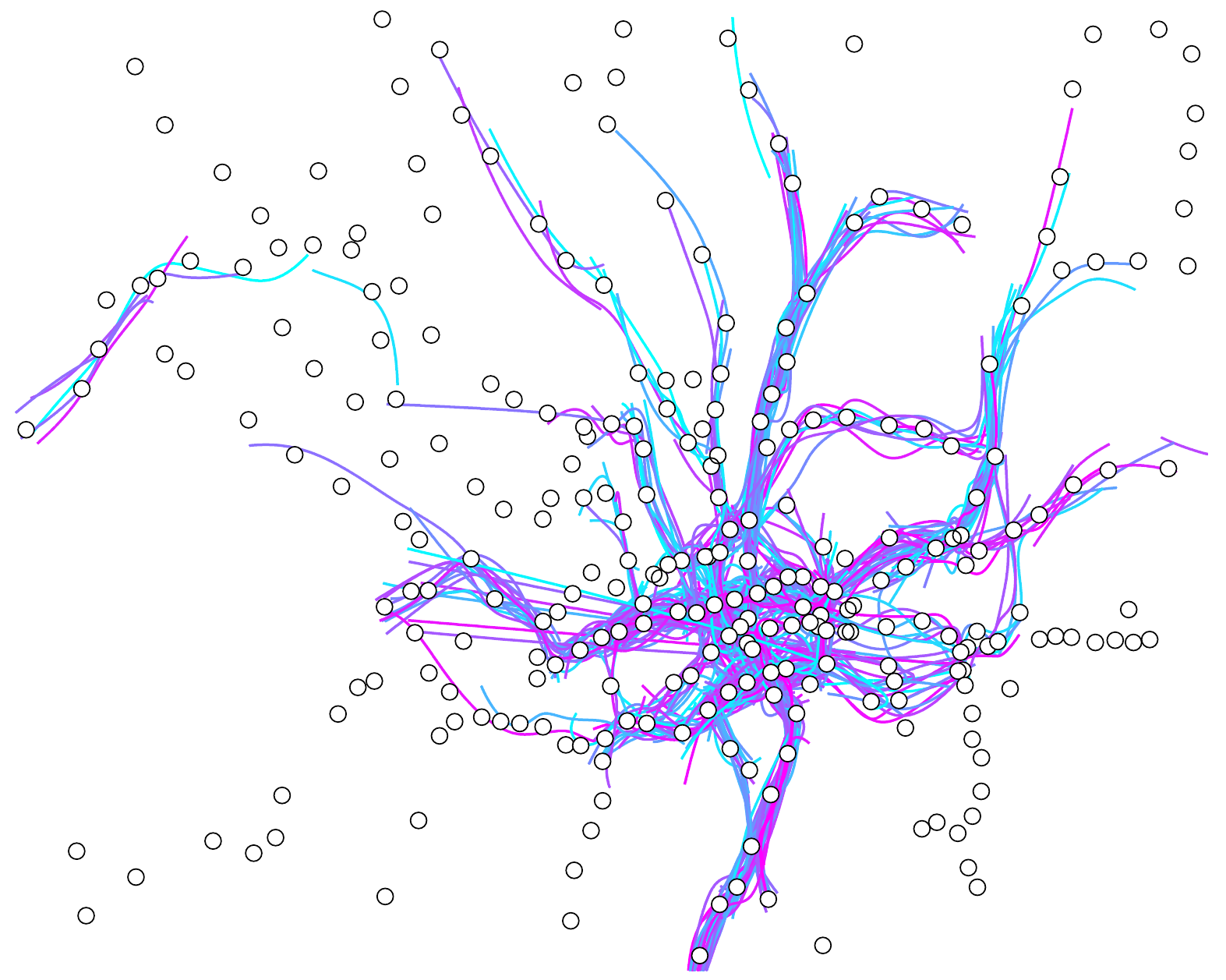}%
\hspace{0.18\columnwidth}%
\includegraphics[width=0.40\textwidth]{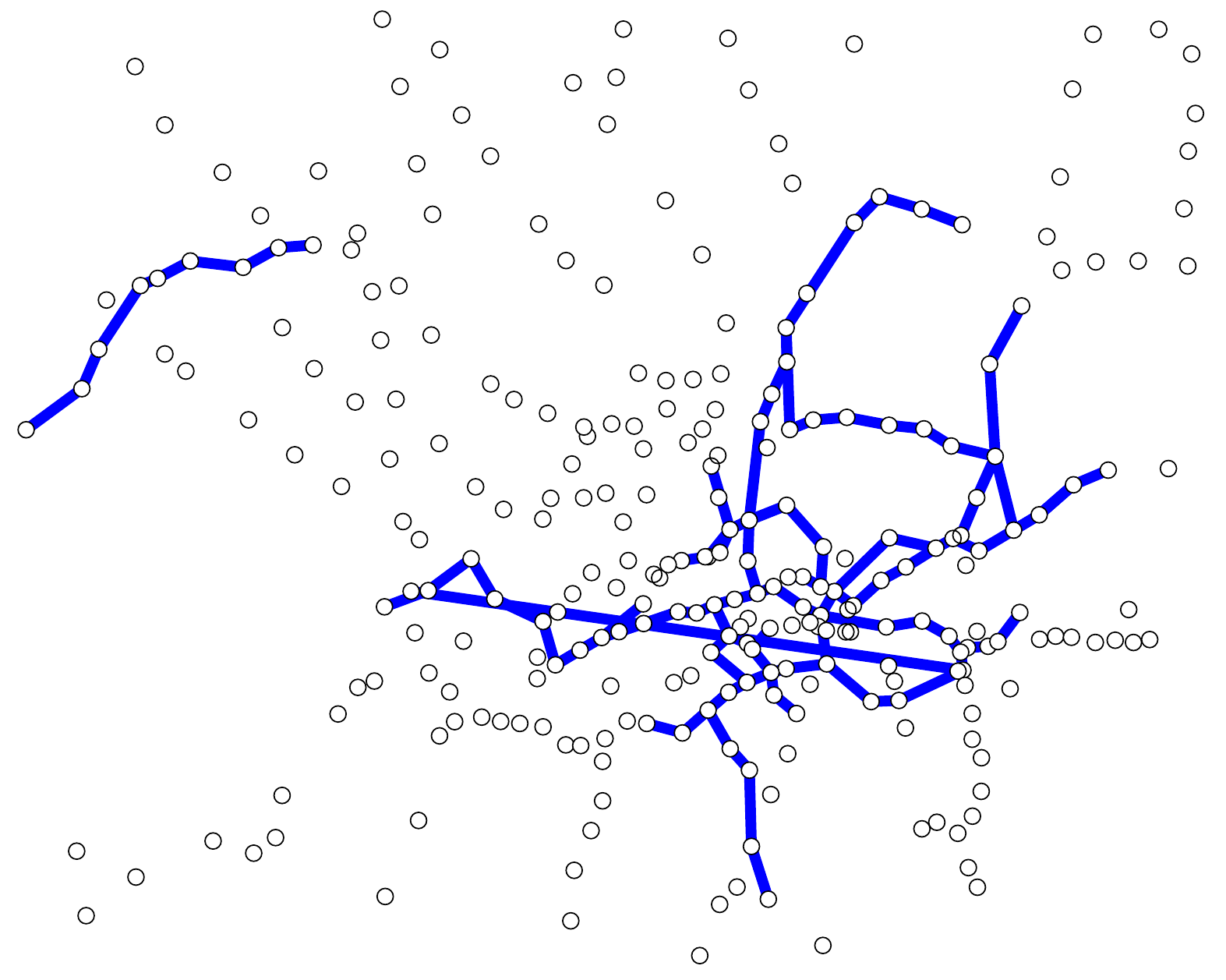}
\end{center}
\caption{\label{london_tube} London tube network, with nodes representing the stations. The figure on the left shows a subset of the trips made,
 and the figure on the right shows the corresponding backbone, as discovered by our
 algorithm. The input data contains only source--destination (indicating start and end points of a trip) pairs and for
 visualization purposes, a B-spline was interpolated along
the shortest path between each such pair. 
The backbone presented on the right covers only 24\% of the edges in the original network and has a stretch factor of 1.58. This means that even with pruning 76\% of the edges in the network, we are able to maintain shortest paths which are at most 1.58 times the shortest path length original graph.}
\label{fig:londontube}
\end{figure*}

The problem we study has applications 
for both {\em exploratory data analysis} and {\em network design}. 
An example application of our algorithm is shown in Figure~\ref{fig:londontube}. Here, we consider a traffic log (Figure~\ref{fig:londontube}, left), which consists of the most popular routes used on the London tube. The backbone produced by our algorithm takes into account this demand (based on the traffic log) and tries to summarize the underlying network, thus presenting us with insights about usage pattern of the London tube (Figure~\ref{fig:londontube}, right).
This representation of the `backbone' of the network could be very useful to identify the important edges to upgrade or to
keep better maintained in order to minimize the total traffic disruptions.




We only consider source-destination pairs in the traffic log, 
and not full trajectories, as 
source-destination information captures
{\em true mobility demand} in a network. 
For example, data about the daily commute from home (source) to office
(destination) is more resilient than trajectory information, which is
often determined by local and transient constraints, like traffic conditions on the road, time of day, etc. 
Furthermore, in 
communication networks, only the source-ip and destination-ip information
is encoded in TCP-IP packets. Similarly, in a city metro, check-in
and check-out information is captured while the intervening movement
is not logged.

The \backbonediscovery\ problem  is an amalgam of
the $k$-{\em spanner} problem~\cite{narasimhan2007geometric}
and the {\em Steiner-forest} problem~\cite{williamson2011design}.
However, our problem formulation will have
elements which are substantially distinct from both of these problems.

In the $k$-spanner problem
the goal is to find a minimum-cost subnetwork \backbone\ of \graph, such
that for {\em each pair} of nodes $u$ and $v$, 
the shortest path between $u$ and $v$ on~\backbone\ 
is at most $k$ times longer than
the shortest path between $u$ and $v$ on~\graph. 
In our problem, we are not necessarily interested in preserving the $k$-factor distance between all nodes but for only a subset of them.

In the Steiner-forest problem
we are given a set of pairs of terminals $\{(\source_i,\dest_i)\}$
and the goal is to find a minimum-cost forest on which each
source $\source_i$ is connected to the corresponding destination~$\dest_i$. Our problem is different from the Steiner-forest problem because we do not need all $\{(\source_i,\dest_i)\}$ to be connected, and try to optimize a stretch factor so that the structural aspect of the network are also taken into account.

A novel aspect of our work is the use of edge-betweenness to guide
the selection of the backbone~\cite{Newman}.  The intuition is as follows. 
An algorithm to solve the Steiner-forest problem will 
try and minimize the sum of cost of edges selected as long as
as the set of terminal 
pairs $\{(\source_{i},\dest_{i})\}$
are connected and is agnostic to minimizing stretch factor. However,
if the edge costs are inversely weighted with {\em edge-betweenness}, 
then edges that can contribute to reducing the stretch factor can be
potentially included into the backbone.

To understand the differences of the proposed \backbonediscovery\ problem
with both the $k$-{\em spanner} and {\em Steiner-forest} formulations, 
consider the example shown in Figure~\ref{counter_example}.
In this example, there are four groups of nodes:
\begin{enumerate}
\item group $A$ consists of $n$ nodes, $a_1,\ldots,a_n$, 
\item group $B$ consists of $n$ nodes, $b_1,\ldots,b_n$,
\item  group $C$ consists of $2$ nodes, $c_1$ and $c_2$, and
\item group $D$ consists of $m$ nodes, $d_1,\ldots,d_m$.
\end{enumerate}
Assume that $m$ is smaller then $n$, and thus much smaller than $n^2$.
All edges shown in the figure have cost $1$, except the edges between
$c_1$ and $c_2$, which has cost~$2$.
Further assume that there is one unit of traffic between each $a_i$ and each~$b_j$, for $i,j=1,\ldots,n$, resulting in $n^2$ source-destination
pairs (the majority of the traffic), 
and one unit of traffic between $d_i$ and $d_{i+1}$, for $i=1,\ldots,m-1$,
resulting in $m-1$ source-destination pairs (some additional marginal
traffic).
The example abstracts a common layout found in many
cities: a few busy centers (commercial, residential, entertainment, etc.)
with some heavily-used links connecting them (group $C$), and some peripheral
ways around, that serve additional traffic (group $D$).

\begin{figure*}[t]

\begin{minipage}{0.31\linewidth}
\centering
\subfloat[]{
	\label{counter_example_a}
	\includegraphics[width=\textwidth, clip=true, trim=5 0 5 5]{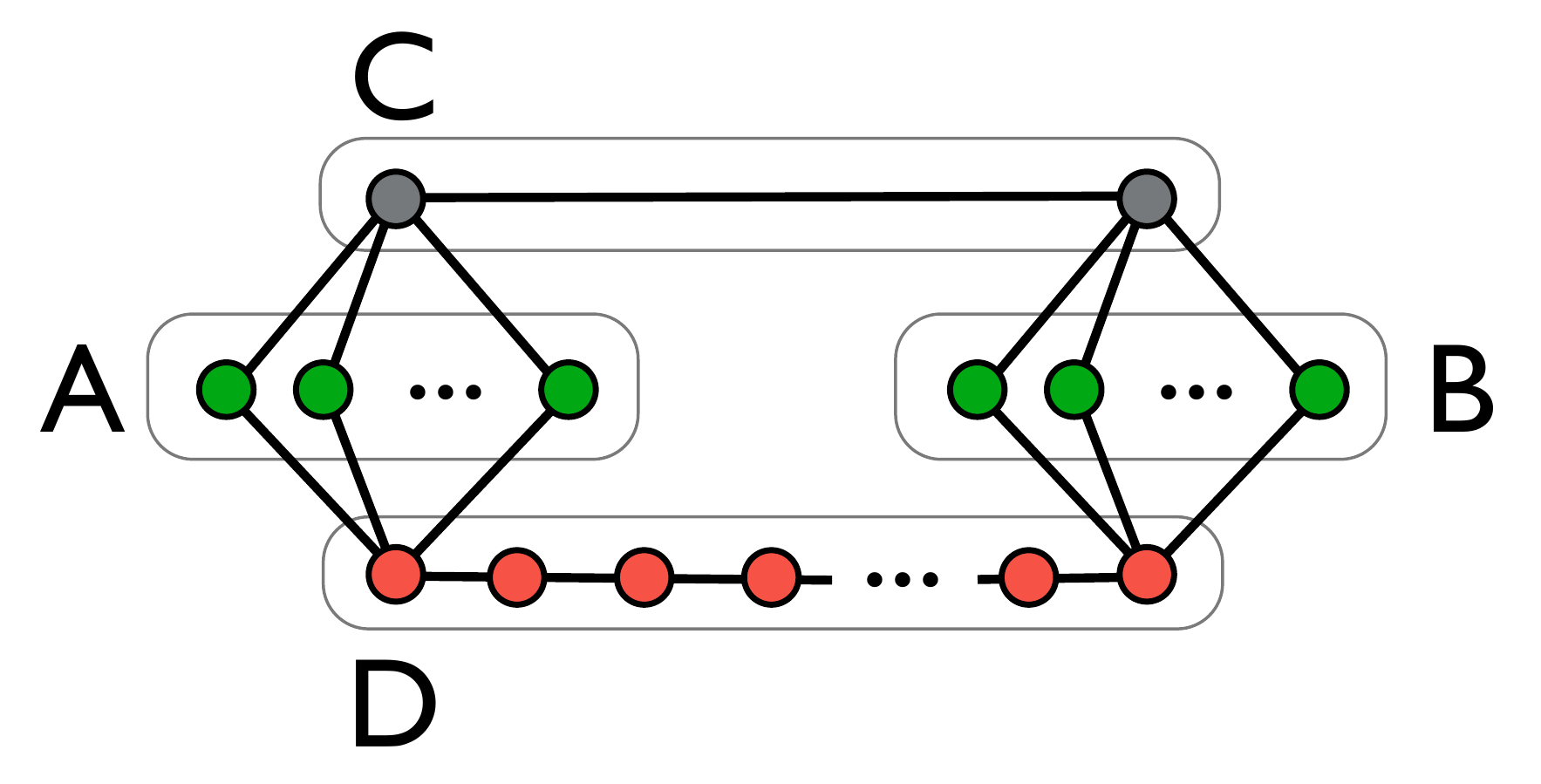}
}
\end{minipage}%
\begin{minipage}{0.31\linewidth}
\centering
\subfloat[]{
	\label{counter_example_b}
	\includegraphics[width=\textwidth, clip=true, trim=5 0 5 5]{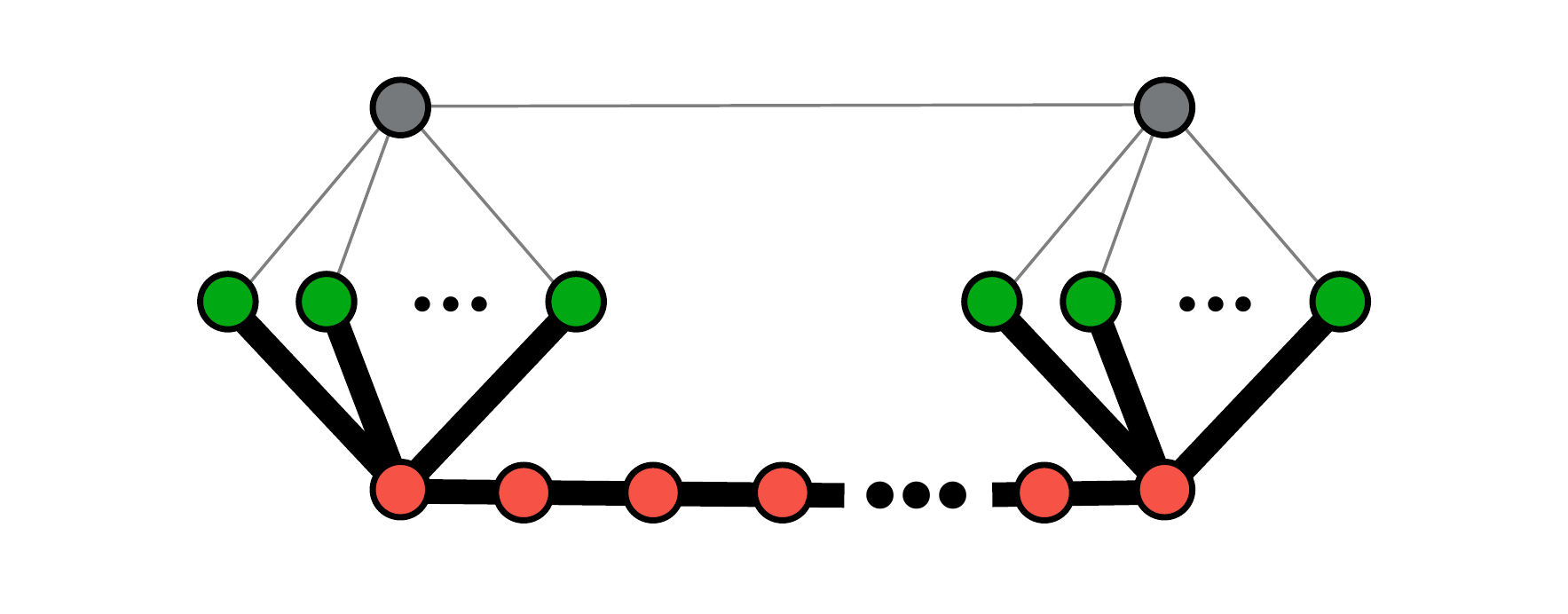}
}
\end{minipage}%
\begin{minipage}{.31\linewidth}
\centering
\subfloat[]{
	\label{counter_example_c}
	\includegraphics[width=\textwidth, clip=true, trim=5 0 5 5]{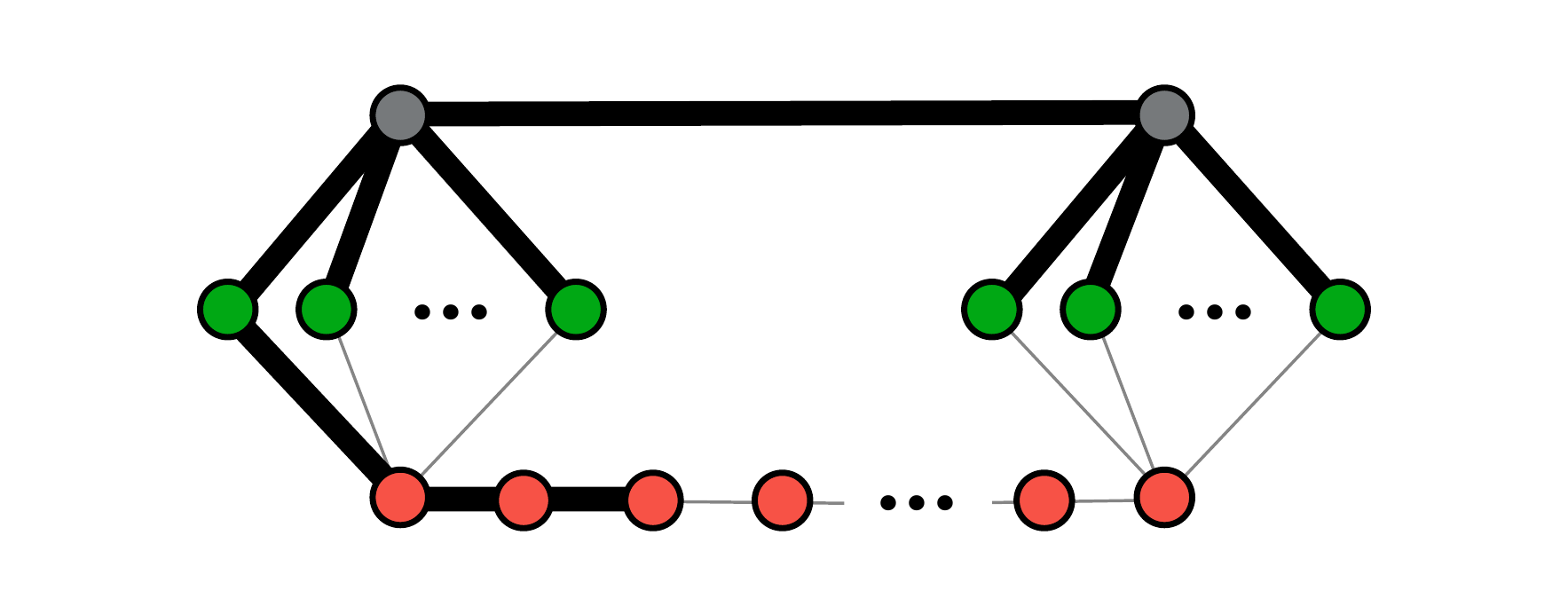}
}
\end{minipage}\par\medskip

\caption{
The \backbonediscovery\ problem solution results in a better network
than the one obtained from the Steiner forest solution.
	(a) A traffic network. We consider
 a unit of traffic from each node in $A$ to each node in $B$, and from each
 node in $D$ to its right neighbor.
 (b) Shown with thick edges is an optimal Steiner forest for certain cost $C$.
 (c) Shown with thick edges is a
 backbone of cost at most $C$ that captures the traffic in the
 network better than the optimal Steiner forest.
}
\label{counter_example}
\end{figure*}

Careful inspection of the above example highlights advantages of
the backbone discovery problem:
\begin{itemize}
\item 
 As opposed to the $k$-spanner problem, we do not need to guarantee
 short paths for all pairs of nodes, but only for those in our traffic log
which makes our approach more general. In particular, based on the budget
requirements a backbone could be designed for the most voluminous 
paths.
\item 
Due to the budget constraint, it may not be possible to
 guarantee connectivity for all pairs in the traffic log. 
We thus need a way to decide which pairs to leave disconnected. 
 Neither the $k$-spanner nor the Steiner-forest problems provision
 for disconnected pairs. 
In fact,
 it is possible that the optimal backbone may even contain cycles
 while leaving pairs disconnected. Again, allowing for a disconnected
backbone, generalizes the Steiner-forest problem and  may help provision for 
a tighter budget.
In order to allow for a disconnected backbone, 
we employ the use of {\em stretch factor}, defined as a 
{\em weighted harmonic mean} over the source-destination pairs of the traffic
log, which provides a principled objective to optimize connectivity while allowing to leave disconnected pairs, when there is insufficient budget. 

\item 
 Certain high cost edges may be an essential part of the backbone that
 other problem formulations may leave out.
 For example, while the edge that connects the nodes in $C$ is a
 very important edge for the overall traffic (as it provides a short route between $A$ and $B$), the
 optimal Steiner-forest solution, shown in
 Figure~\ref{counter_example_b}, prefers the long path along the
 nodes in~$D$. Our algorithm includes the component $C$ (as seen in Figure~\ref{counter_example_c}) because it is an edge that has a high edge-betweenness.

\end{itemize}


The rest of the paper is organized as follows. 
In Section~\ref{sec:problem}, we rigorously define the \backbonediscovery\ problem.
Section~\ref{sec:algorithms} introduces our algorithm based on the greedy approach, 
Section~\ref{sec:experiments} details our experimental evaluation, results and discussion.
In Section~\ref{sec:related} we survey related work and distinguish our problem formulation with other relevant approaches. 
We conclude in Section~\ref{sec:conclusions} with a summary and potential directions for future work.

\section{Problem definition}
\label{sec:problem}

Let $\graph=(\vertices,\edges)$ be a network, with
$|\vertices|=\novertices$ and $|\edges|=\noedges$. 
For each edge $\edge\in\edges$ there is a cost $\ccost(\edge)$. 
Additionally, we consider a traffic log {\calL}, 
specified as a set of triples $(\source_i, \dest_i, \traffic_i)$, 
with $\source_i,\dest_i\in\vertices$, and with $i=1,\ldots,\nopairs$.
A triple $(\source_i, \dest_i, \traffic_i)$ indicates the fact that
$\traffic_i$ units of traffic have been recorded between nodes
$\source_i$ and $\dest_i$. 

We aim at discovering the {\em backbone} of traffic networks. 
A backbone $\backbone$ is a subset of the edges of the network $\graph$,
that is, $\backbone\subseteq\edges$ that provides
a good summarization for the whole
traffic in~\calL.
In particular, we require that if the available traffic had used
only edges in the backbone \backbone, it should have been almost as
efficient as using all the edges in the network. We formalize this intuition below. 

Given two nodes $\source,\dest\in\vertices$ 
and a subset of edges $\aset\subseteq\edges$, 
we consider the shortest path 
$\dist_\aset(\source,\dest)$
from \source\ to \dest\
that uses only edges in the set \aset.
In this shortest-path definition, edges are counted according to their cost \ccost.
If there is no path from $\source$ to $\dest$ using only edges in
\aset, we define $\dist_\aset(\source,\dest)=\infty$.
Consequently, 
$\dist_\edges(\source,\dest)$
is the shortest path from $\source$ to $\dest$ using all the edges in
the network, and 
$\dist_\backbone(\source,\dest)$
is the shortest path from $\source$ to $\dest$ using only edges in the
backbone~\backbone.

To measure the quality of a backbone \backbone,
with respect to some traffic log 
$\calL=\{(\source_i, \dest_i, \traffic_i)\}$
we use the concept of {\em stretch factor}. 
Intuitively, we want to consider shortest paths from $\source_i$ to $\dest_i$, 
and evaluate how much longer are those paths on the backbone
\backbone, 
than on the original network. The idea of using stretch factor for evaluating the quality of a subgraph has been used 
extensively in the past in the context of spanner graphs~\cite{narasimhan2007geometric}.

In order to aggregate shortest-path information for all
source--destination pairs in our log in a
meaningful way, we need to address two issues. 
The first issue is that not all source--destination pairs have the
same volume in the traffic log. 
This can be easily addressed by weighting the contribution of
each pair $(\source_i,\dest_i)$ by its corresponding
volume~$\traffic_i$.

The second issue is that since we aim at discovering very sparse
backbones, many source--destination pairs in the log could be
disconnected in the backbone. 
To address this problem we aggregate shortest-path
distances using the {\em harmonic mean}.
This idea, which has been proposed by Marchiori and
Latora~\cite{marchiori2000harmony} and has also been used by Boldi and
Vigna~\cite{boldi2013axioms} in measuring centrality in networks, 
provides a very clean way to deal with infinite distances: 
if a source--destination pair is not
connected, their distance is infinity, so the harmonic mean accounts
for this by just adding a zero term in the summation.
Using the arithmetic mean is problematic, as we would need to add an 
infinite term with other finite numbers. 


Overall, given a set of edges $\aset\subseteq\edges$, we measure
the connectivity of the traffic log 
$\calL = \{(\source_i, \dest_i, \traffic_i)\}$,  $|\calL| = k$ by
\[
\HM_\calL(\aset) = 
\left( \sum_{i=1}^k \traffic_i \right)
\left( 
 \sum_{i=1}^k \frac{\traffic_i}{\dist_\aset(\source_i,\dest_i)}
\right)^{-1}.
\]
The {\em stretch factor} of a backbone \backbone\ is then defined
as 
\[
\sfactor_\calL(\backbone)=
\frac{\HM_\calL(\backbone)}{\HM_\calL(\edges)}.
\]
The stretch factor is always greater or equal than 1. 
The closer it is to 1, the better the connectivity that it offers
to the traffic log~\calL.

\smallskip
We are now ready to formally define the problem of backbone discovery.
\begin{problem}[{\backbonediscovery}]
\label{problem:optimal-backbone-design}	
Consider a network $\graph=(\vertices,\edges)$ and a traffic log $\calL=\{ (\source_i, \dest_i, \traffic_i)\}$.
Consider also a cost budget \budget.
The goal is to find a backbone network $\backbone\subseteq\edges$ of total cost {\budget} 
that minimizes the stretch factor $\sfactor_\calL(\backbone)$
or report that no such solution exists. 
\end{problem}	

As one may suspect, \backbonediscovery\ is an \NPhard\ problem. 

\begin{lemma}
\label{lemma:np-hardness}
The \backbonediscovery\ problem, defined in Problem~\ref{problem:optimal-backbone-design}, is \NPhard.
\end{lemma}
\begin{proof}[Sketch]
We obtain a reduction from the \setcover\ problem: 
given a ground set $\universe=\{\e_1,\ldots,\e_\noe\}$, 
a collection $\collection = \{\set_1,\ldots,\set_\nos\}$ 
of subsets of \universe, and an integer~\bk, 
determine whether there are \bk\ sets in \collection\
that cover all the elements of \universe.

Given an instance of the \setcover\ problem we form an instance of the
\backbonediscovery\ problem as follows (See Figure~\ref{fig:np_gadget} for illustration). 
We create one node $\e_i$ for each $\e_i\in\universe$
and one node $\setnode_j$ for each $\set_j\in\collection$.
We also create a special node $\terminal$.
We create an edge $(\e_i,\setnode_j)$ if and only if $\e_i\in\set_j$
and we assign to this edge cost 0.
We also create an edge $(\setnode_j,\terminal)$ for all
$\set_j\in\collection$ and we assign to this edge cost 1.
As far as the traffic log is concerned, 
we set $\calL=\{(\e_i,\terminal,1)\mid\e_i\in\universe\}$, 
that is, we pair each $\e_i\in\universe$ with the special node
\terminal\ with volume~1.
Finally we set the budget $\budget=\bk$.
It is not difficult to see that the instance of the
\backbonediscovery\ problem constructed in this way has a solution
with stretch factor 1 if and only if the given instance of the
\setcover\ problem has a feasible solution. \hfill
\end{proof}

\begin{figure}[h]
\centering
\includegraphics[width=0.5\textwidth]{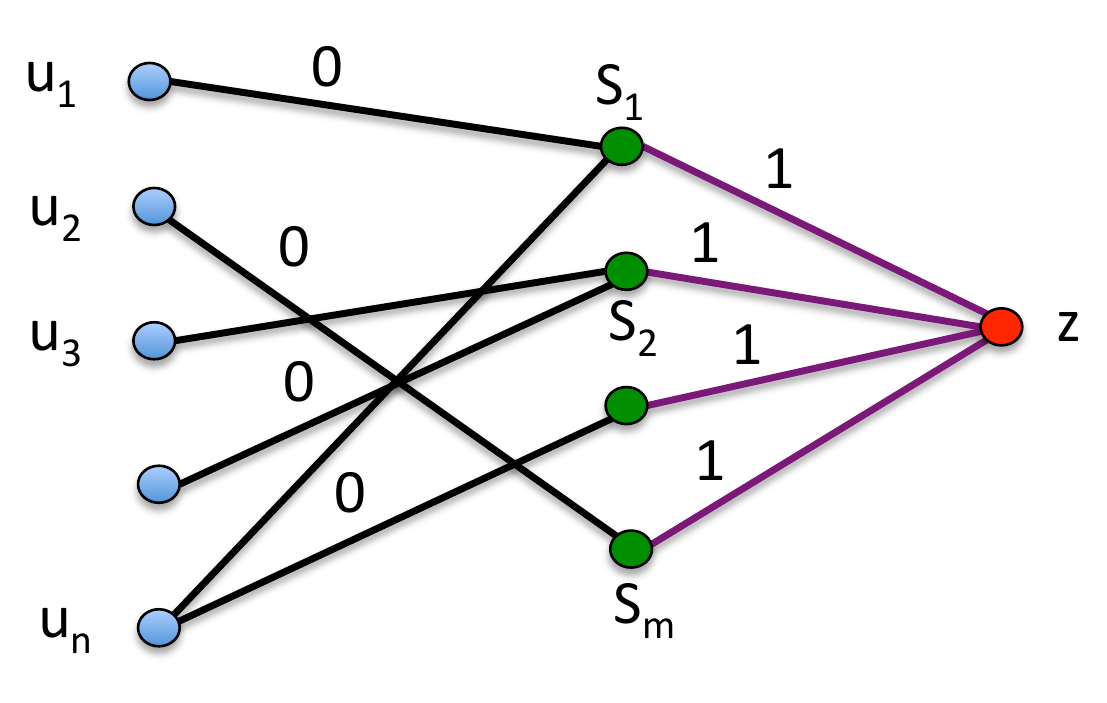}
\caption{Reduction from Set Cover to \backbonediscovery\ for the
log $\calL=\{(u_{i},z,1)|u_{i} \in U \}$}
\label{fig:np_gadget}
\end{figure}

\section{Algorithm}
\label{sec:algorithms}

The algorithm we propose for the \backbonediscovery\ problem 
is a {\em greedy} heuristic that connects one-by-one the
source--destination pairs of the traffic log~\calL. 
A distinguishing feature of our algorithm is that it utilizes a notion
of {\em edge benefit}.
In particular, we assume that for each edge $\edge\in\edges$ we have
available a benefit measure $\benefit(\edge)$. 
The higher is the measure $\benefit(\edge)$ the more beneficial it is
to include the edge \edge\ in the final solution. 
The benefit measure is computed using the traffic log \calL\ and it
takes into account the global structure of the network \graph.

The more central an edge is with respect to a traffic log, 
the more beneficial it is to include it in the solution, 
as it can be used to serve many source--destination pairs. 
In this paper we use {\em edge-between\-ness} as a centrality measure,
adapted to take into account the traffic log.
We also experimented with {\em commute-time centrality}, 
but edge-between\-ness was found to be more effective. 

Our algorithm relies on the notion of {\em effective distance}
$\effective(\edge)$, 
defined as $\effective(\edge)=\ccost(\edge)/\benefit(\edge)$, 
where $\ccost(\edge)$ is the cost of an edge~$\edge\in\edges$, 
and $\benefit(\edge)$ is the edge-between\-ness of \edge.
The intuition is that by dividing the cost of each edge by its
benefit, we are biasing the algorithm towards selecting edges with
high benefit.

We now present our algorithm in more detail.

\subsection{The greedy algorithm.}
\label{section:greedy}

As discussed above, our algorithm operates with effective distances 
$\effective(\edge)=\ccost(\edge)/\benefit(\edge)$, 
where $\benefit(\edge)$ is a benefit score for each edge~\edge. 
The objective is to obtain a cost/benefit trade-off: 
edges with small cost and large benefit are favored to be included in
the backbone. 
In the description of the greedy algorithm that follows, we assume that the
effective distance $\effective(\edge)$ of each edge is given as
input. 

The algorithm works in an iterative fashion, maintaining and growing
the backbone, starting from the empty set.
In the $i$-th iteration the algorithm picks a source--destination pair 
$(\source_i,\dest_i)$ from the traffic log \calL, 
and ``serves'' it. 
Serving a pair $(\source_i,\dest_i)$ means computing a
shortest path $\ppath_i$ from $\source_i$ to $\dest_i$, 
and adding its edges in the current \backbone, if they are not already
there.
For the shortest-path computation the algorithm uses the effective
distances~$\effective(\edge)$.
When an edge is newly added to the backbone
its cost is subtracted from the available budget. Here, the actual cost of the edge $\ccost(\edge)$ (instead of the $\effective(\edge)$) is used. 
Also its effective distance is reset to zero, 
since it can be used for free in subsequent iterations of the algorithm.
The source--destination pair that is chosen to be served in each
iteration is the one that reduces the stretch factor the most at that iteration; 
and hence the greedy nature of the algorithm. 

The algorithm proceeds until it exhausts all its budget or until the
stretch factor becomes equal to 1
(which means that all pairs in the log are served via a shortest
path).
The pseudo-code for the greedy algorithm is shown as
Algorithm~\ref{algorithm:greedy}.

\begin{algorithm}[t]
\caption{\label{algorithm:greedy}The greedy algorithm}
\begin{algorithmic}[1]
\REQUIRE{Network $\graph=(\vertices,\edges)$, edge costs
 $\ccost(\edge)$, benefit costs $\benefit(\edge)$, 
 cost budget \budget, traffic log $\calL = \{ (\source_i, \dest_i, \traffic_i)\}$ }
\ENSURE{ A subset of edges $\backbone\subseteq\edges$ of total cost
 $\ccost(\backbone)\le\budget$ and small stretch factor
 $\sfactor(\backbone)$}
\FOR {all $\edge\in\edges$}
\STATE {$\effective(\edge)\leftarrow \ccost(\edge)/\benefit(\edge)$}
\ENDFOR
\STATE {$\backbone\leftarrow\emptyset$}
\STATE {$\sfactor\leftarrow\infty$}

\WHILE{$(\budget>0)$ and $(\sfactor>1)$}
\FOR{ each $(\source_i, \dest_i, \traffic_i)\in\calL$ } 
\STATE {$\ppath_i\leftarrow${\sc ShortestPath}$(\source_i,\dest_i,\graph,\effective)$}
\STATE {$\sfactor_i\leftarrow${\sc StretchFactor}$(\backbone \cup \ppath_i,\graph,\calL)$}
\ENDFOR
\STATE {$p^*\leftarrow \min_i \{ \sfactor_i \}$}
\hfill\COMMENT{the path with min. stretch factor in the above iteration} 
\STATE {$\backbone' \leftarrow p^*\setminus\backbone$}
\hfill\COMMENT{edges to be newly added} 
\IF {$\ccost(\backbone')>\budget$}
\STATE {Return $\backbone$} 
\hfill\COMMENT{budget exhausted} 
\ENDIF
\STATE {$\backbone\leftarrow\backbone \cup \backbone'$}
\hfill\COMMENT{add new edges in the backbone } 
\STATE {$\effective(\backbone')\leftarrow 0$}
\hfill\COMMENT{reset cost of newly added edges} 
\STATE {$\budget\leftarrow\budget-\ccost(\backbone')$}
\hfill\COMMENT{decrease budget} 
\STATE {$\sfactor\leftarrow${\sc StretchFactor}$(\backbone,\graph,\calL)$}
\hfill\COMMENT{update $\sfactor$} 
\ENDWHILE
\STATE {Return $\backbone$}
\end{algorithmic}
\end{algorithm}

We are experimenting with two variants
of this greedy scheme, depending on the benefit score we use.
These are the following:
\begin{description}
\item[\greedy:] We use uniform benefit scores ($\benefit(\edge)=1$). 
\item[\greedyeb:] The benefit score of an edge is set equal to its 
 {\em edge-between\-ness centrality}. 
\end{description}


\subsection{Speeding up the greedy algorithm.}
\label{sec:optimizations}

As we show in the experimental section the greedy algorithm provides
solutions of good quality, in particularly the variant with the
edge-betweenness weighting scheme. 
However, the algorithm is computationally expensive, and thus in this
section we discuss a number of optimizations. 
We start by analyzing the running time of the algorithm.

\spara{Running time.}
Assume that the benefit scores $\benefit(\edge)$ are given for all edges
$\edge\in\edges$, and that the algorithm performs \iter\ iterations until it
exhausts its budget. 
In each iteration we need to perform $\bigO(\nopairs^2)$ shortest-path
computations, where \nopairs\ is the size of the traffic log \calL.
A shortest path computation is 
$\bigO(\noedges+\novertices\log \novertices)$, and thus the overall
complexity of the algorithm is 
$\bigO(\iter\nopairs^2(\noedges+\novertices\log \novertices))$.
The number of iterations \iter\ depends on the available budget and in
the worst case it can be as large as \nopairs. 
However, since we aim at finding sparse backbones, 
 the number of iterations is typically smaller.

\spara{Optimizations with no approximation.}
We first show how to speed up the algorithm, 
while guaranteeing the same solution with 
the na\"{i}ve implementation of the greedy.
Since the most expensive component is the computation of
shortest paths on the newly-formed network, we make sure that we
compute the shortest path only when needed.
Our optimizations consist of two parts.

First, during the execution of the algorithm we 
maintain the connected components that the backbone creates  in the network.
Then, we do not need to compute shortest paths for all
$(\source_i,\dest_i)$ pairs, for which $\source_i$ and $\dest_i$
belong to different connected components; we know that their distance
is $\infty$.
This optimization is effective at the early stages of the algorithm,
when many terminals belong to different connected components. 

Second, when computing the decrease in the stretch factor due to a candidate
shortest path to be added in the backbone, 
for pairs for which we have to recompute a shortest-path distance, 
we first compute an optimistic lower bound, based on the shortest path
on the whole network (which we compute once in a preprocessing step). 
If this optimistic lower bound is not better than the current best
stretch factor then we can skip the computation of the shortest path
on the backbone. 

As shown in the empirical evaluation of our algorithms, depending on
the dataset, these optimization heuristics lead to 20--35\% improvement
in performance.

\spara{Optimization based on landmarks.}
To further improve the running-time of the algorithm
we compute shortest-path distances using
landmarks~\cite{DBLP:conf/wsdm/SarmaGNP10,DBLP:conf/cikm/PotamiasBCG09},
an effective technique to approximate distances on graphs. 
Here we use the approach of Potamias et
al.~\cite{DBLP:conf/cikm/PotamiasBCG09}:
A set of \nolandmarks\ landmarks 
$\landmarks = \{\landmark_1,\ldots,\landmark_\nolandmarks\}$
is selected and for each vertex $\vertex\in\vertices$ the distances 
$\dist(\vertex,\landmark_i)$
to all landmarks are computed and stored as an
$\nolandmarks$-dimensional vector representing vertex~\vertex.
The distance between two vertices $\vertex_1,\vertex_2$ 
is then approximated as 
$\min_{i} \{ \dist(\vertex_1,\landmark_i) +
\dist(\vertex_2,\landmark_i) \}$, 
i.e., the tightest of the upper bounds induced by the set of landmarks
\landmarks.

To select landmarks we use high-degree non-adjacent vertices in the
graph, which is reported as one of the best performing methods by
Potamias et al.~\cite{DBLP:conf/cikm/PotamiasBCG09}.
Note that the distances are now approximations to the true distances,
and the method trades accuracy by efficiency via the number of
landmarks selected. 
In practice a few hundreds of landmarks are sufficient to provide
high-quality approximations even for graphs with millions of
vertices~\cite{DBLP:conf/cikm/PotamiasBCG09}.

For the running-time analysis, 
note that in each iteration we need to compute the distance between all graph
vertices and all landmarks.
This can be done with \nolandmarks\ single-source shortest-path
computations and the running time is
$\bigO(\nolandmarks(\noedges+\novertices\log\novertices))$.
The landmarks can then be used to approximate shortest-path distances
between all source-destination pairs in the traffic log \calL, 
with running time $\bigO(\nopairs\nolandmarks)$.
Thus, the overall complexity is 
$\bigO(\iter\nolandmarks(\nopairs+\noedges+\novertices\log\novertices))$.
Since \nolandmarks\ is expected to be much smaller than \nopairs, 
the method provides a significant improvement over the na\"{i}ve
implementation of the greedy. As shown in the experimental evaluation, using landmarks provides an improvement of up to 4 times in terms of runtime in practice.



\subsection{Edge-between\-ness centrality.}
\label{section:ebc}

As we already discussed in the previous sections, our
greedy algorithm uses edge centrality
measures for benefit scores $\benefit(\edge)$.
In this section we discuss edge
between\-ness, and in particular show how we adapt the measure to take into
account the traffic log~\calL.

We first recall the standard definition of edge-betweennes.
Given a network $\graph=(\vertices,\edges)$, we define $\vertexpairs=
{\vertices \choose 2}$ to be the set of all pairs of nodes of
$\graph$.
Given a pair of nodes $(\source, \dest)\in \vertexpairs$ and an edge $\edge\in\edges$, 
we define by 
$\nsp_{\source,\dest}$ the total number of shortest paths from $\source$ to $\dest$, 
and by
$\nsp_{\source,\dest}(\edge)$ the total number of shortest paths from $\source$ to $\dest$ that pass though edge \edge.

The standard definition of edge-between\-ness centrality $\EBC(\edge)$
of edge {\edge} is the following:
\[
\EBC(\edge) =
\sum_{(\source,\dest)\in\vertexpairs} \frac{\nsp_{\source,\dest}(\edge)}{\nsp_{\source,\dest}}.
\]
In our problem setting we take into account the traffic log
$\calL=\{(\source_i, \dest_i, \traffic_i)\}$, and we define the
edge-between\-ness $\EBClog(\edge)$ of an edge \edge\ with respect to
log~\calL, as follows.
\[
\EBClog(\edge) =
\sum_{(\source,\dest,\traffic)\in\calL} \traffic\,\frac{\nsp_{\source,\dest}(\edge)}{\nsp_{\source,\dest}}.
\]
In this modified definition only the source--destination pairs of the
log \calL\ contribute to the centrality of the edge \edge, and the
amount of contribution is proportional to the corresponding traffic.
The adapted edge-between\-ness can still be
computed in $\bigO(\novertices\noedges)$ time~\cite{Brandes2007Centrality}.

\section{Experimental evaluation}
\label{sec:experiments}

The aim of the experimental section is to evaluate the performance of
the proposed algorithm, the optimizations, and the edge-betweenness measure.
We also compare our algorithm with other state-of-the-art methods which attempt to solve a similar problem.

\spara{Datasets.}
We experiment with six real-world datasets, 
four transportation networks, one web network and one internet-traffic network.
For five of the datasets we also obtain real-world traffic,
while for one we use synthetically-generated traffic.
The characteristics of our datasets are provided in
Table~\ref{tab:datasets}, and a brief description follows.

\begin{table}
\centering
\resizebox{\linewidth}{!}{%
\begin{tabular}{lcrrcc}
\hline
\multirow{2}{*}{Dataset} & \multirow{2}{*}{Type} & 
\multirow{2}{*}{\# Nodes} & \multirow{2}{*}{\# Edges} & Real & Real \\
 & & & & network & traffic \\
\hline
\london & transportation & 316 & 724 & \checkmark & \checkmark \\
\flights & transportation & 1\,268 & 51\,098 & \checkmark & \checkmark \\
\ukroad & transportation & 8\,341 & 13\,926 & \checkmark & - \\
\nyctaxi & transportation & 50\,736 & 158\,898 & \checkmark & \checkmark \\
\wikispeedia & web & 4\,604 & 213\,294 & \checkmark & \checkmark \\ 
\abeline & internet & 12 & 15 & \checkmark & \checkmark \\
\hline
\end{tabular}}
\caption{\label{tab:datasets}Dataset statistics.}
\end{table}

\spara{\london.} 
The London Subway network consists of subway stops and links between
them.\footnote{{\small\url{http://bit.ly/1C9PbLT}}}
We use the geographic distance between stations as a proxy for edge
costs.
We also obtain a traffic log $\calL$ extracted from the Oyster card system.\footnote{{\small\url{http://bit.ly/1qM2BYi}}}
The log consists of aggregate trips made by passengers 
between pairs of stations during a one-month period (Nov-Dec 2009). 
We filter out source-destination pairs with traffic less than 100 and
remove bi-directional pairs by selecting one of them at random and
summing up their traffic. 

\spara{\flights}.
We obtain a large network of US airports, and the list of all flights
within the US during 2009--2013, from the Bureau of Transportation Statistics.\footnote{{\small\url{http://1.usa.gov/1ypXYvL}}}
Flying distances between airports, obtained using Travelmath.com,
are used as edge costs. 
The traffic volume is the number of flights between airports.

\spara{\nyctaxi}.
We obtain the complete road network of NYC using OpenStreetMap data.\footnote{{\small\url{http://metro.teczno.com/\#new-york}}}
In this network each node corresponds to a road intersection and
each link corresponds to a road segment. Edge costs are computed
as the geographic distances between the junctions.
Data on the pickup and drop-off locations of NYC taxis for 2013 was 
used to construct the traffic
log.\footnote{{\small\url{http://chriswhong.com/open-data/foil_nyc_taxi/}}} 
The $2\,000$ most 
frequently used source-destination pairs was used to create the
traffic log.

\spara{\wikispeedia}.
Wikispeedia\footnote{{\small\url{http://snap.stanford.edu/data/wikispeedia.html}}}~\cite{west2009wikispeedia} is an online crowd sourcing game designed to measure semantic distances between 2 wikipedia pages using the paths taken by humans to reach from one page to the other. This dataset contains a base network of hyperlinks between Wikipedia pages and the paths taken by users between two pages. 
We construct the traffic log using the unique (start, end) pages from this data.

\spara{\ukroad}.
Next we consider the UK road
network.\footnote{{\small\url{http://www.dft.gov.uk/traffic-counts/download.php}}}
The network construction is similar to what was done with the \nyctaxi\ data.
For simplicity we use only the largest connected component. 
Since we were not able to obtain real-world traffic data for this network,
we generate synthetic traffic logs $\calL$ simulating different
scenarios. 
In particular we generate traffic logs according to four different
distributions:
($i$) power-law traffic volume, power-law \source-\dest\ pairs; 
($ii$) power-law traffic volume, uniformly random \source-\dest\ pairs; 
($iii$) uniformly random traffic volume, power-law \source-\dest\ pairs; and
($iv$) uniformly random traffic volume, uniformly random \source-\dest\ pairs.
These different distributions capture different traffic volume possibility and hence help in understanding the behavior of our algorithm with respect to the traffic log $\calL$.

\spara{\abeline}.
For a qualitative analysis we also consider the well known \abeline\ dataset
consisting of a sample of the network traffic extracted from the
Internet2 backbone\footnote{{\small\url{http://www.internet2.edu}}} 
and
that carries the  network traffic between major universities in the
continental US. 
The network consists of twelve nodes and 15 high-capacity links. 
Associated with each physical link, we also have capacity of the link which
serves as a proxy for the cost of the link.
We obtain traffic logs from 2003 between all pairs of nodes.

\spara{Baseline.}
To obtain better intuition for the performance of our methods
we define a simple baseline, where 
a backbone is created by adding edges in increasing order of their
effective distances $\effective(\edge) =
\ccost(\edge)/\benefit(\edge)$, 
where $\benefit(\edge)$ is edge-betweenness; 
this was the best-performing baseline among other baselines we
tried, such as adding source--destination pairs one by one (i) randomly, (ii) in decreasing order of volume ($w_i$), (iii) in increasing order of effective distance defined using closeness centrality, etc.

\begin{figure*}[t]
\begin{minipage}{.24\linewidth}
\centering
\subfloat[]{\label{}\includegraphics[width=0.95\textwidth, height=\textwidth]{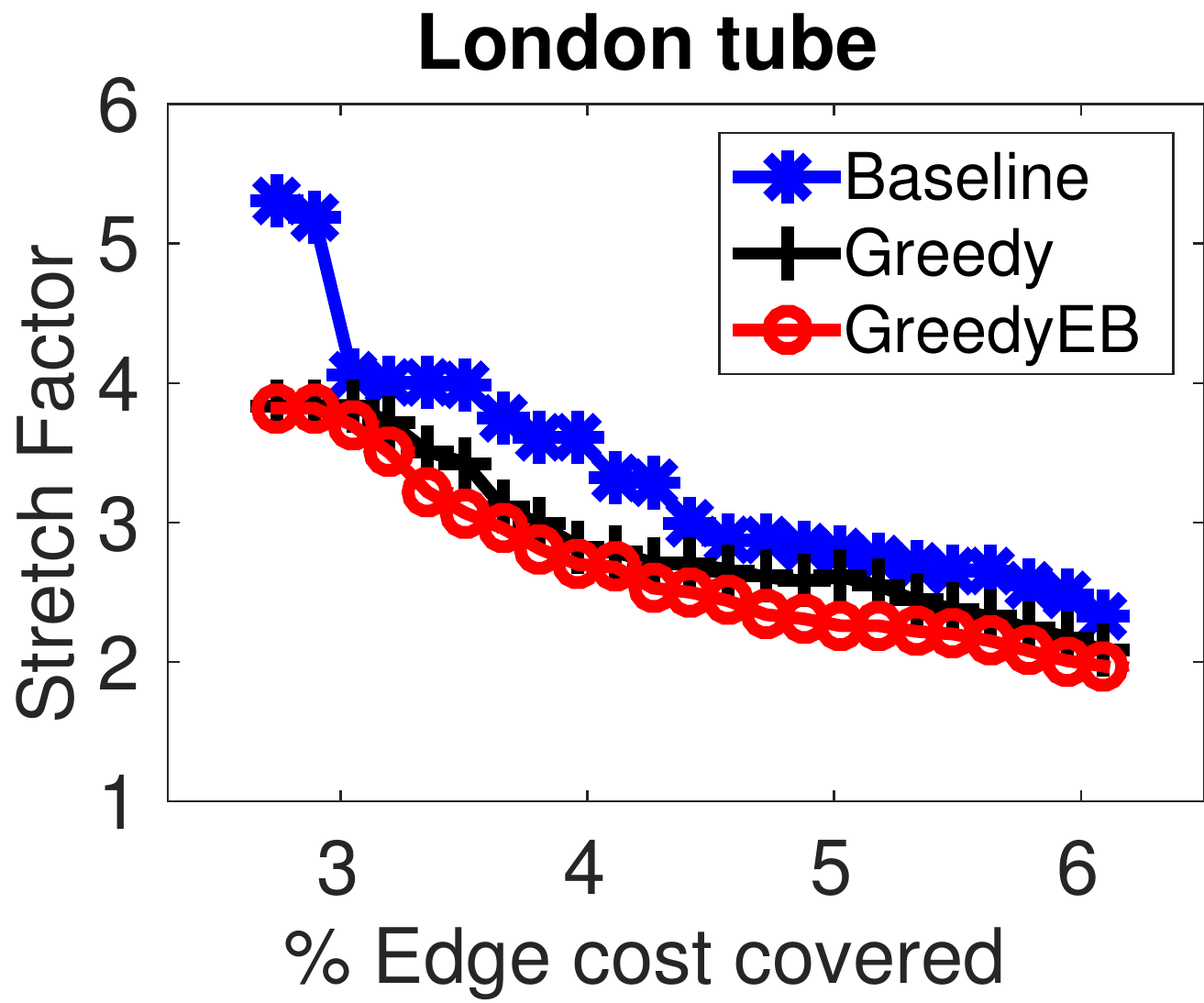}}
\end{minipage}%
\begin{minipage}{.24\linewidth}
\centering
\subfloat[]{\label{}\includegraphics[width=0.95\textwidth, height=\textwidth]{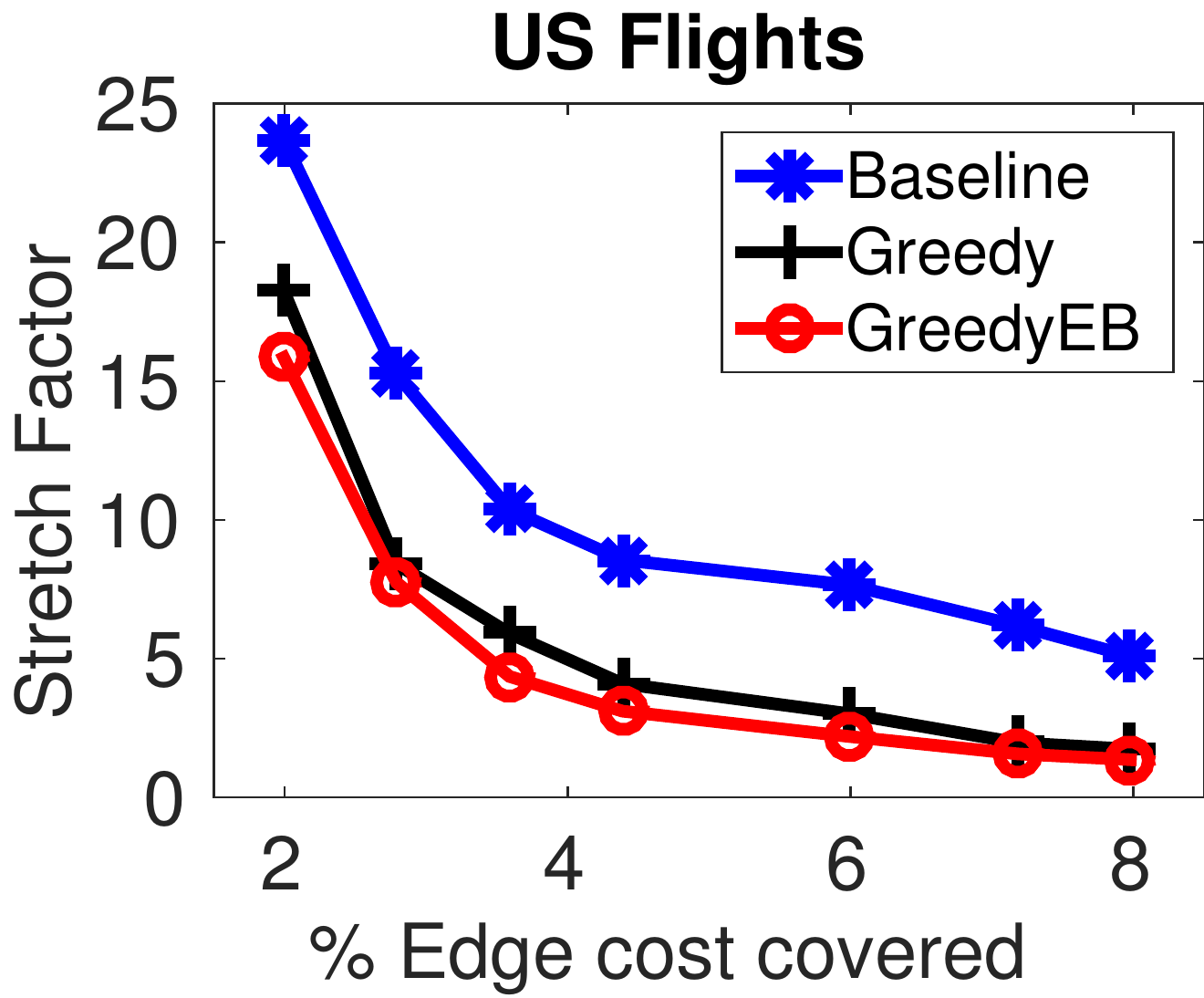}}
\end{minipage}
\begin{minipage}{.24\linewidth}
\centering
\subfloat[]{\label{}\includegraphics[width=0.95\textwidth, height=\textwidth]{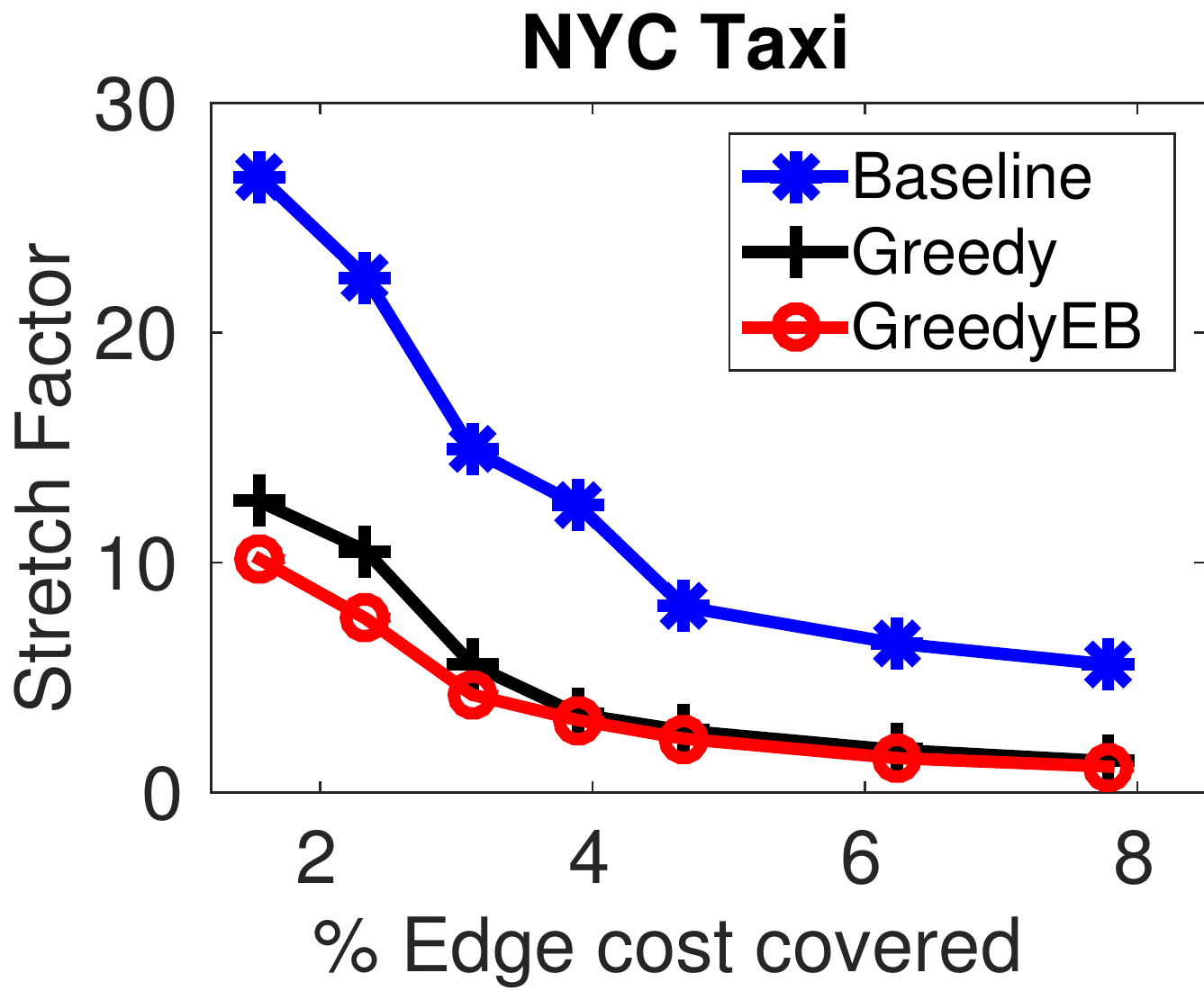}}
\end{minipage}
\begin{minipage}{.24\linewidth}
\centering
\subfloat[]{\label{}\includegraphics[width=0.95\textwidth, height=\textwidth]{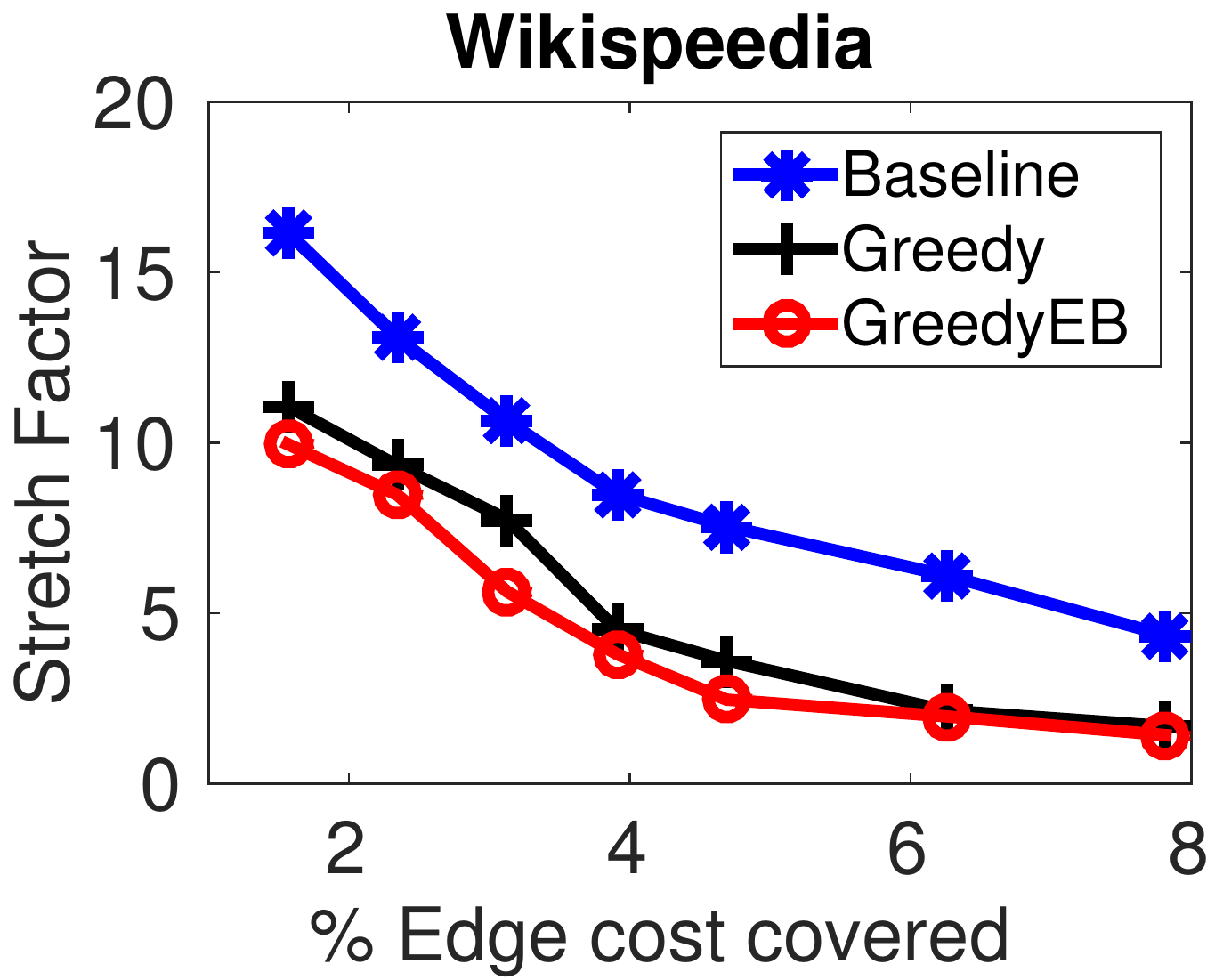}}
\end{minipage}\par\medskip
\begin{minipage}{.24\linewidth}
\centering
\subfloat[]{\label{}\includegraphics[width=\textwidth, height=\textwidth]{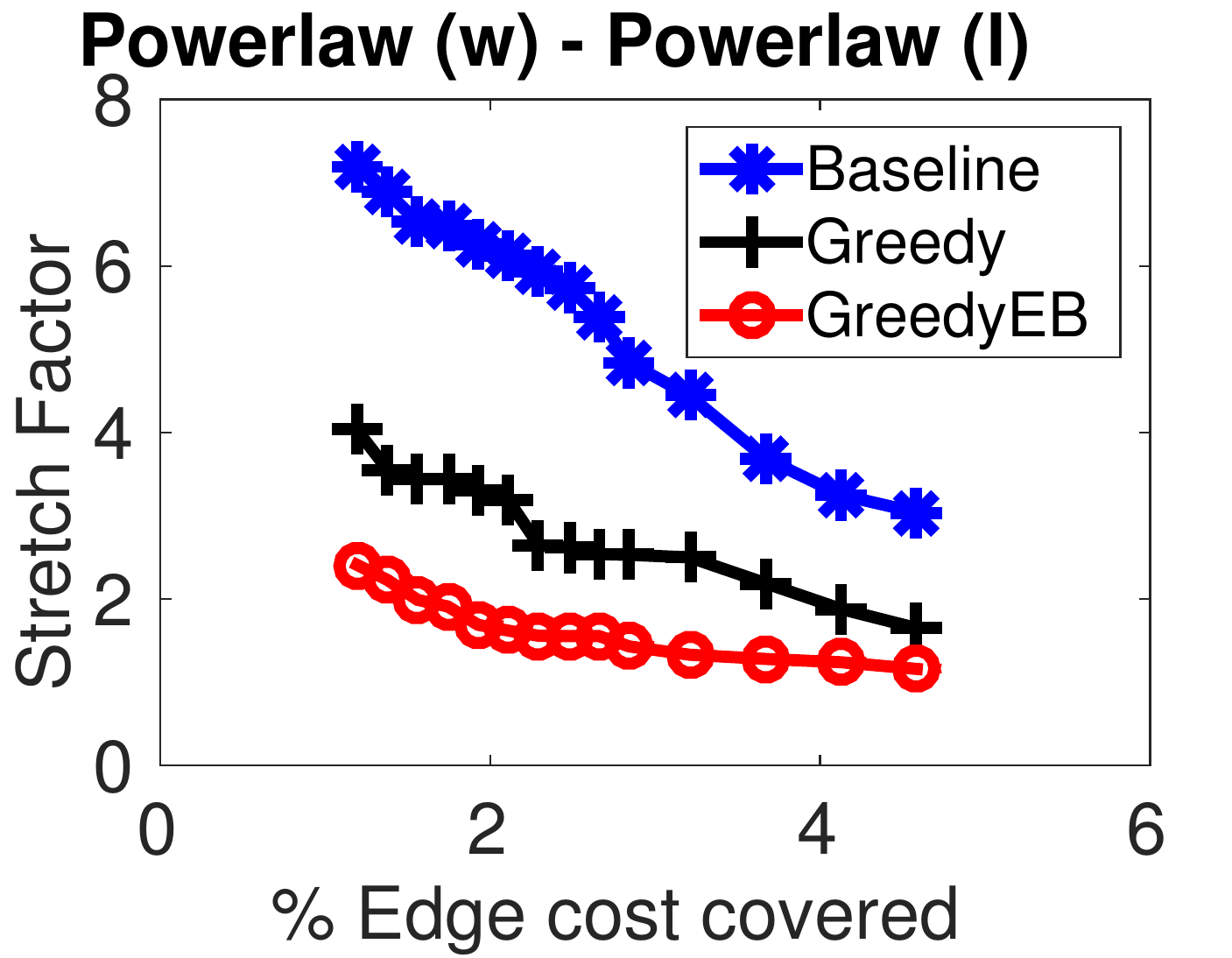}}
\end{minipage}
\begin{minipage}{.24\linewidth}
\centering
\subfloat[]{\label{}\includegraphics[width=\textwidth, height=\textwidth]{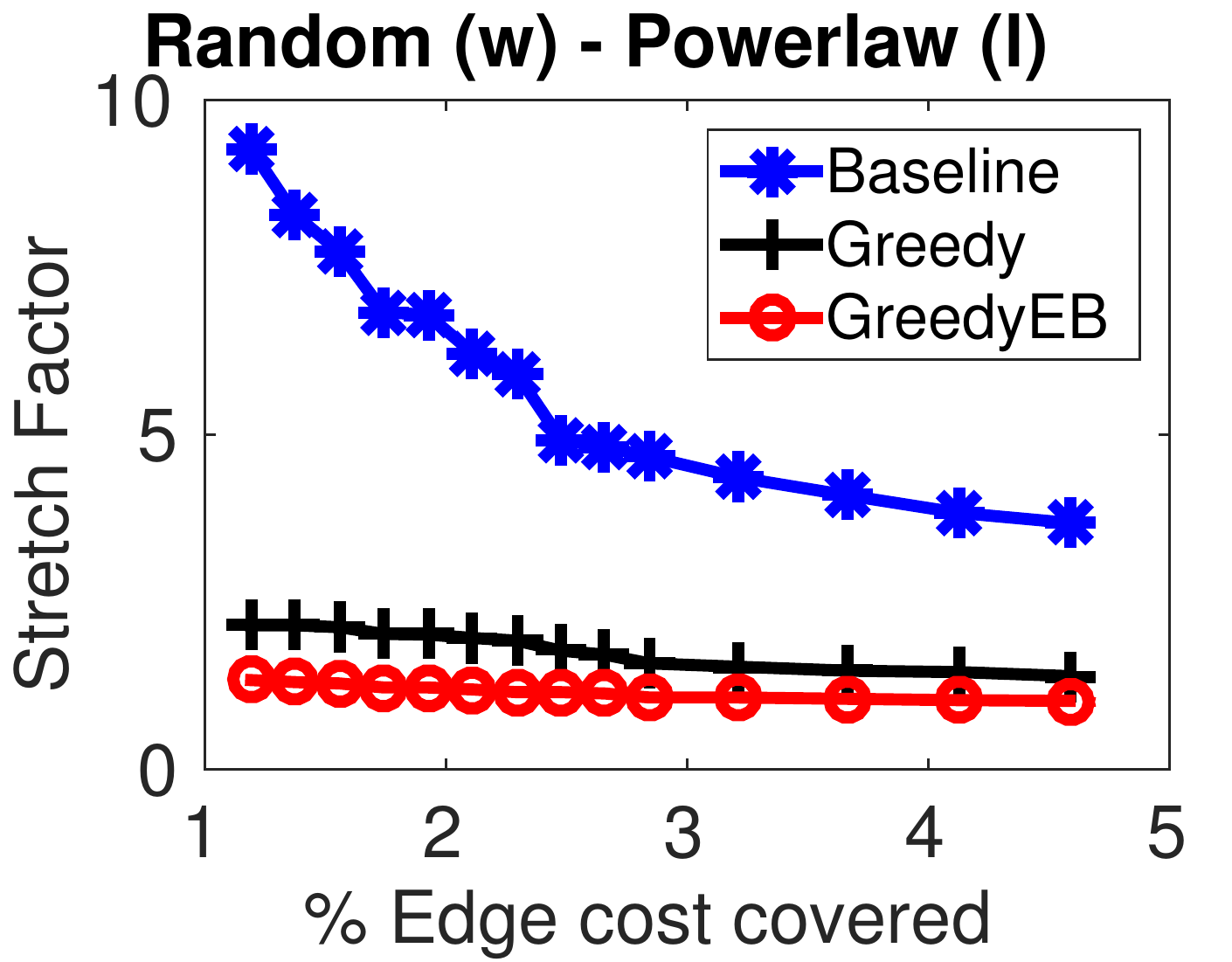}}
\end{minipage}
\begin{minipage}{.24\linewidth}
\centering
\subfloat[]{\label{}\includegraphics[width=\textwidth, height=\textwidth]{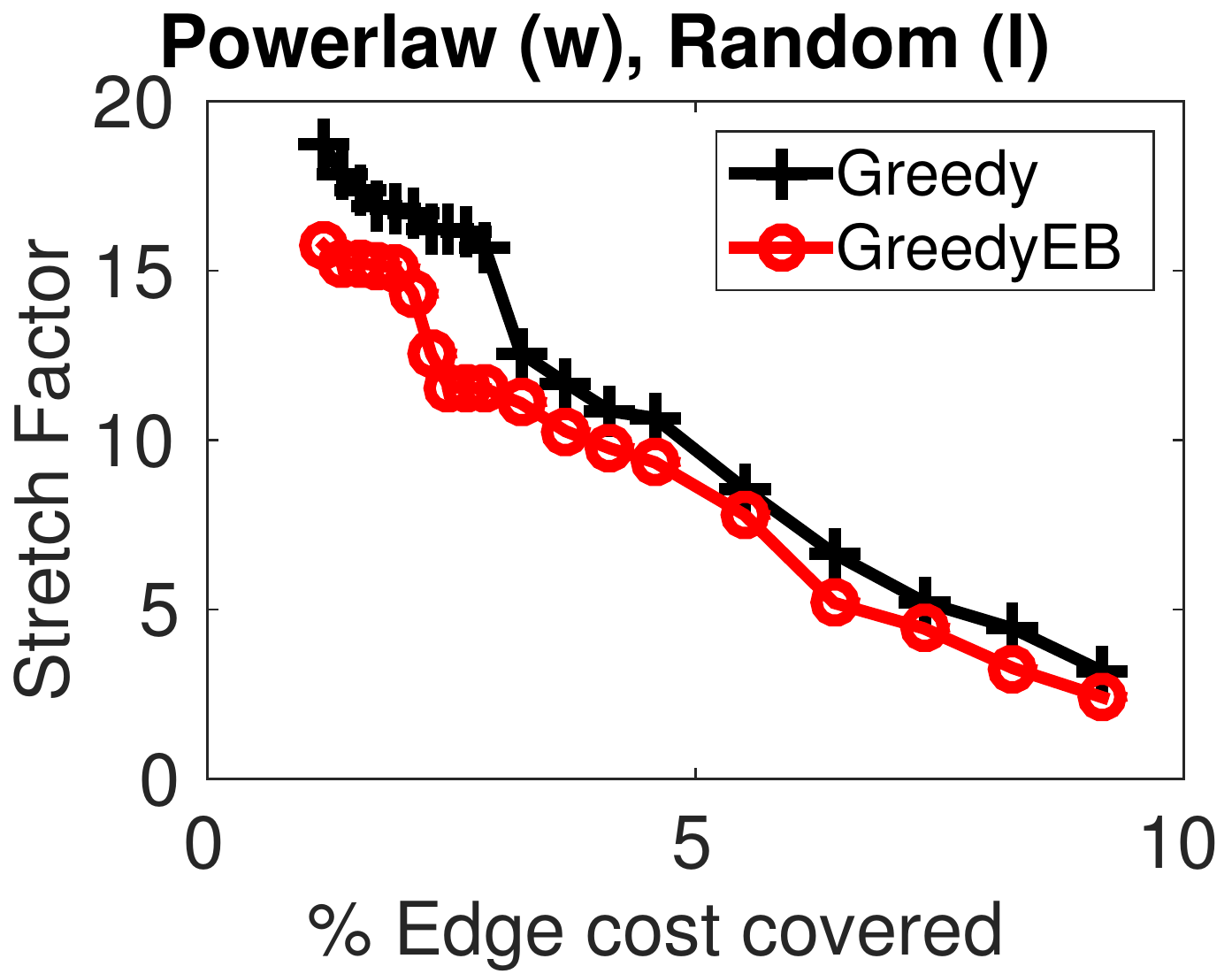}}
\end{minipage}
\begin{minipage}{.24\linewidth}
\centering
\subfloat[]{\label{}\includegraphics[width=\textwidth, height=\textwidth]{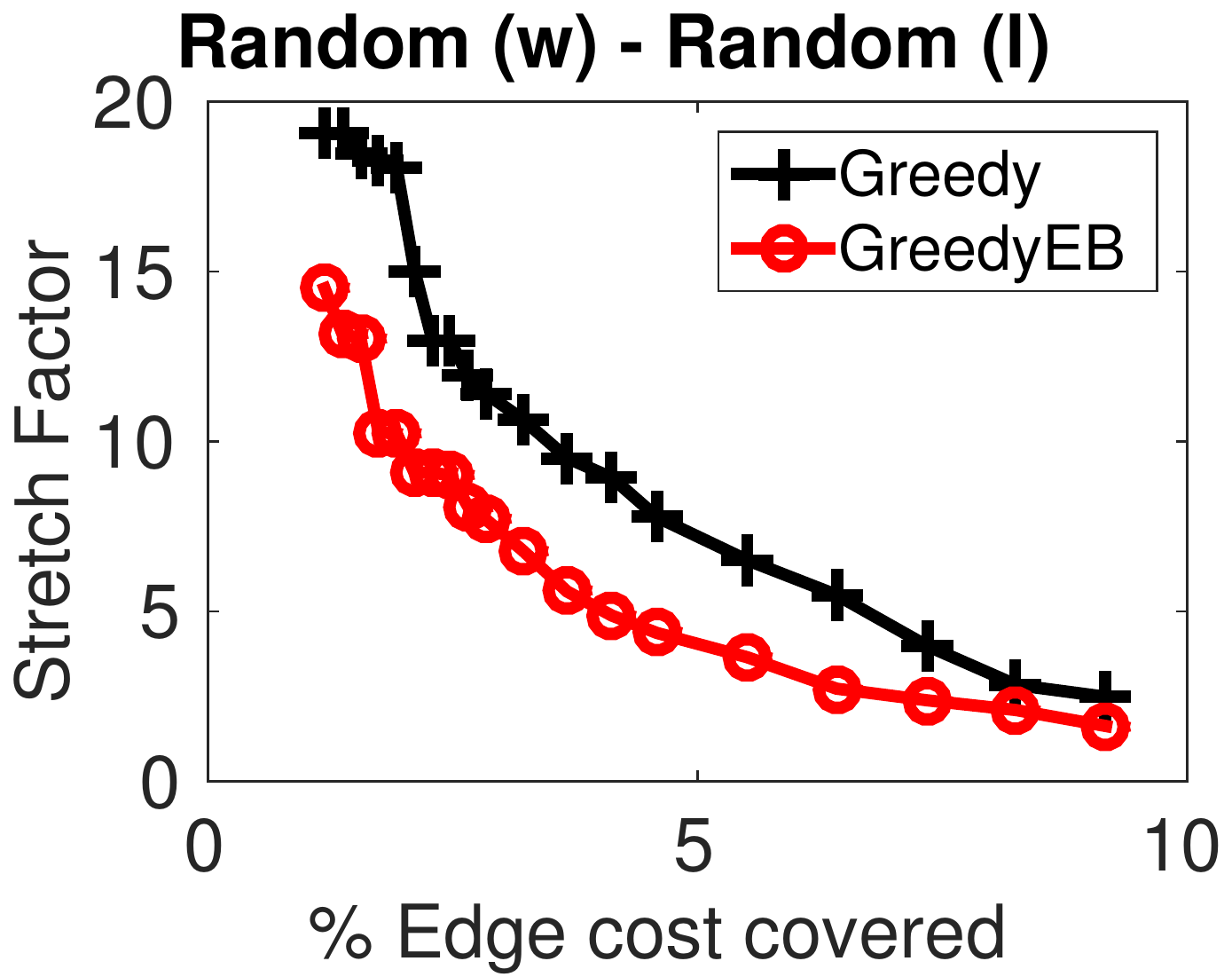}}
\end{minipage}\par\medskip

\caption{Effect of edge-betweenness on the performance of the Greedy algorithm, for various datasets (a) \london, (b) \flights, (c) \nyctaxi, (d) \wikispeedia, (e--h) \ukroad. Baseline is missing in figures (g) and (h) because the stretch factor was very large or infinity. We see a consistent trend that using edge-betweenness improves the performance. In Figures (e--h), $(w)$ indicates traffic volume, and $(l)$ indicates the log.}
\label{fig:ebVSnormal}

\end{figure*}

\subsection{Quantitative results.}
%
We focus our evaluation on three main criteria: 
(i) Comparison of the performance with and without the edge-betweenness measure;
(ii) effect of the optimizations, in terms of quality and speedup; and
(iii) effect of allocating more budget on the stretch factor. 

\spara{Effect of edge-betweenness.}
We study the effect of using edge-betweenness in the Greedy algorithm.
The results are presented in Figure~\ref{fig:ebVSnormal}.

\spara{Effect of landmarks.} 
Landmarks provide faster computation with a trade off for quality. 
Figure~\ref{fig:landmarks_timetaken}
shows the speedup achieved when using landmarks.
In the figures, BasicGreedyEB indicates the greedy algorithm that doesn't use any optimizations. GreedyEBCC makes use of the optimizations proposed in Section~\ref{sec:optimizations} which do not use approximation. GreedyEBLandmarks* makes use of the landmarks optimatization and the * indicates the number of landmarks we tried.
Figure~\ref{fig:landmarks_performance} shows the performance of \greedyeb\ algorithm with and without using landmarks.

\spara{Budget vs.\ stretch factor.}
We examine the trade-off between budget and stretch factor for our
algorithm and its variants. A lower stretch factor for the same budget indicates that the algorithm is able to pick better edges for the backbone.
Figure~\ref{fig:ebVSnormal} shows the trade-off between
budget and stretch factor for all our datasets.
In all figures the budget used by the algorithms, shown in the
$x$-axis, is expressed as a percentage of the total cost of all the
edges in the network. 


\begin{figure*}[t]
\begin{minipage}{.24\linewidth}
\centering
\subfloat[]{\label{}\includegraphics[width=\textwidth, height=\textwidth]{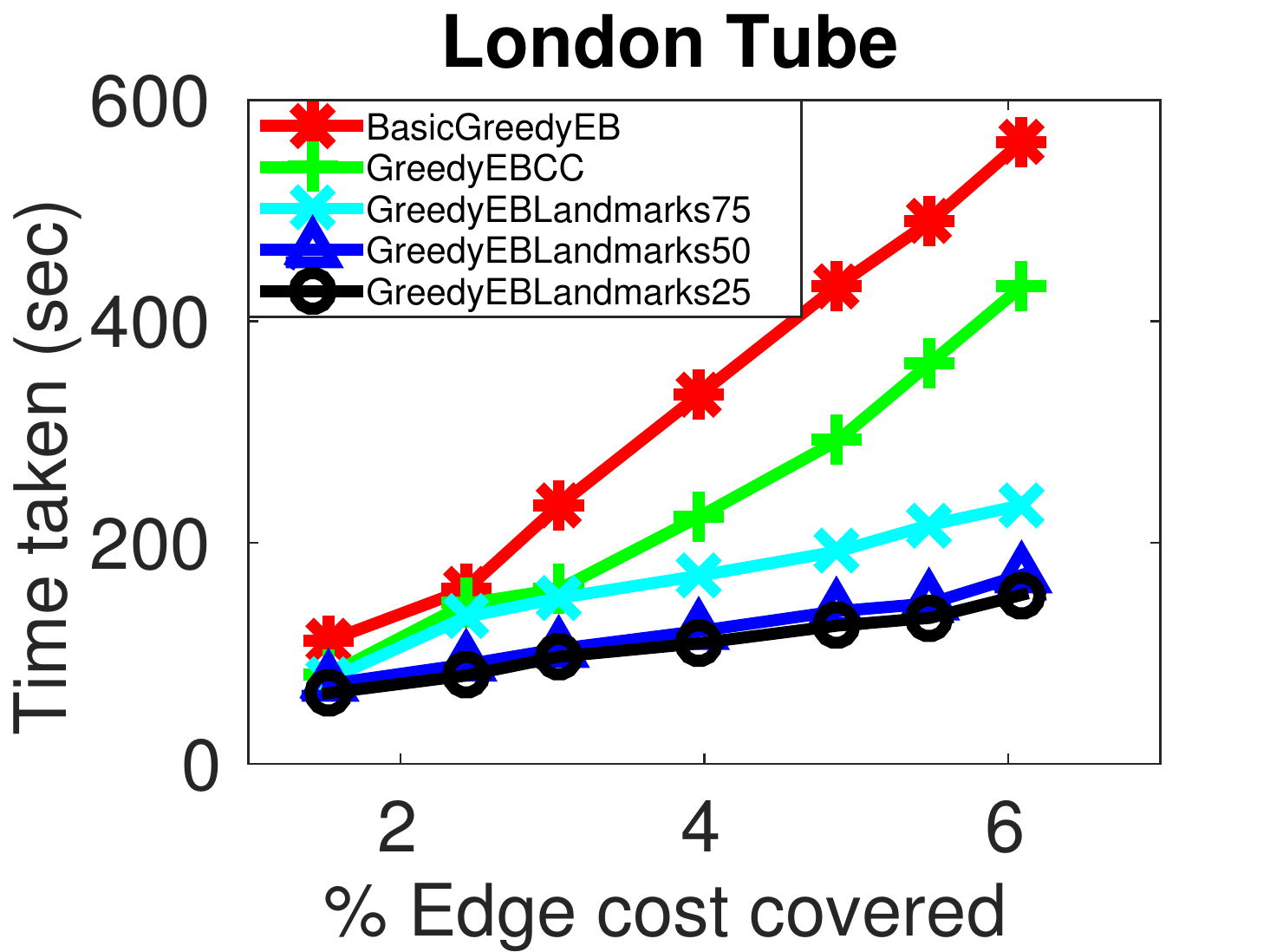}}
\end{minipage}%
\begin{minipage}{.24\linewidth}
\centering
\subfloat[]{\label{}\includegraphics[width=\textwidth, height=\textwidth]{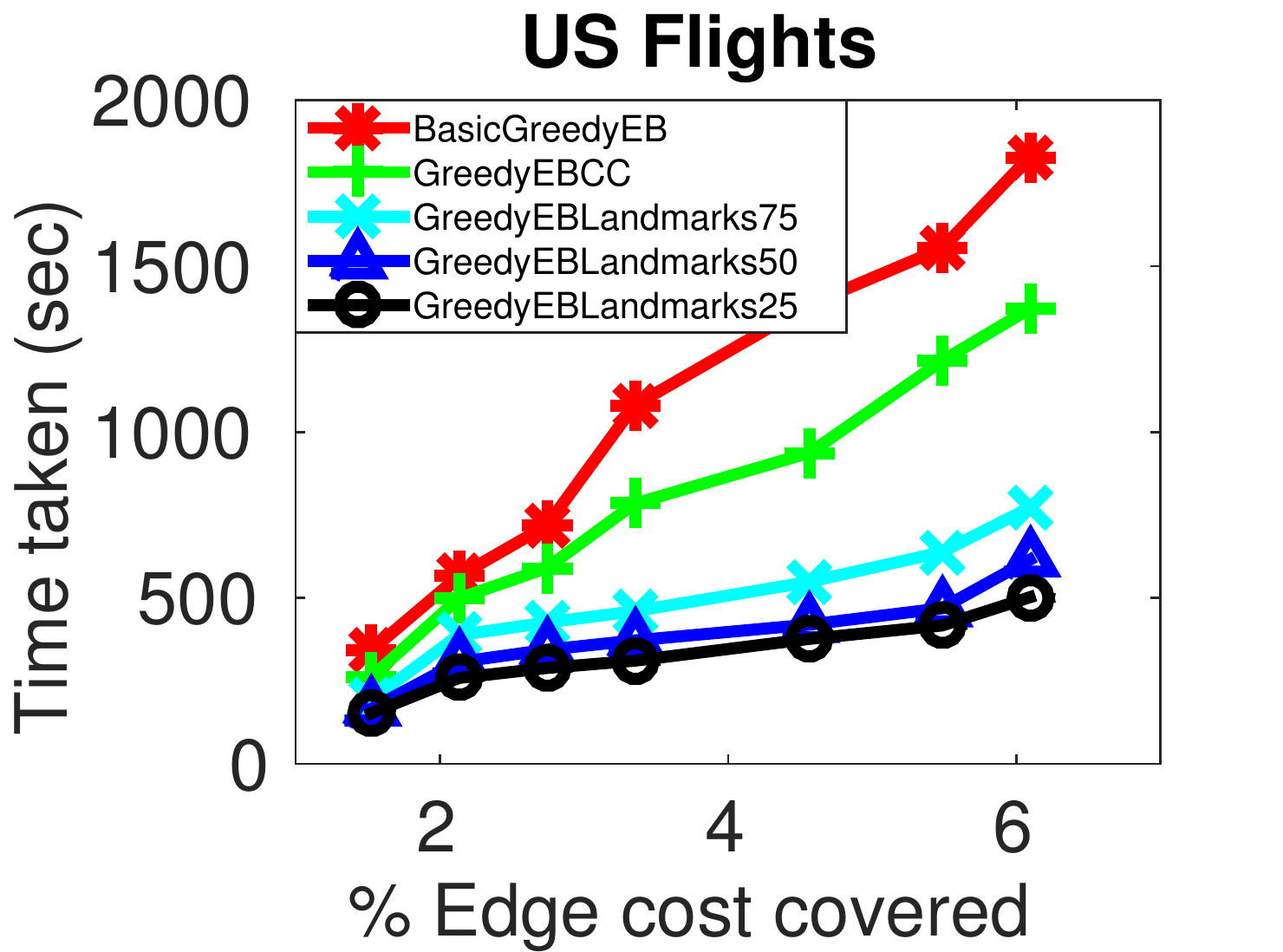}}
\end{minipage}
%
\begin{minipage}{.24\linewidth}
\centering
\subfloat[]{\label{}\includegraphics[width=\textwidth, height=\textwidth]{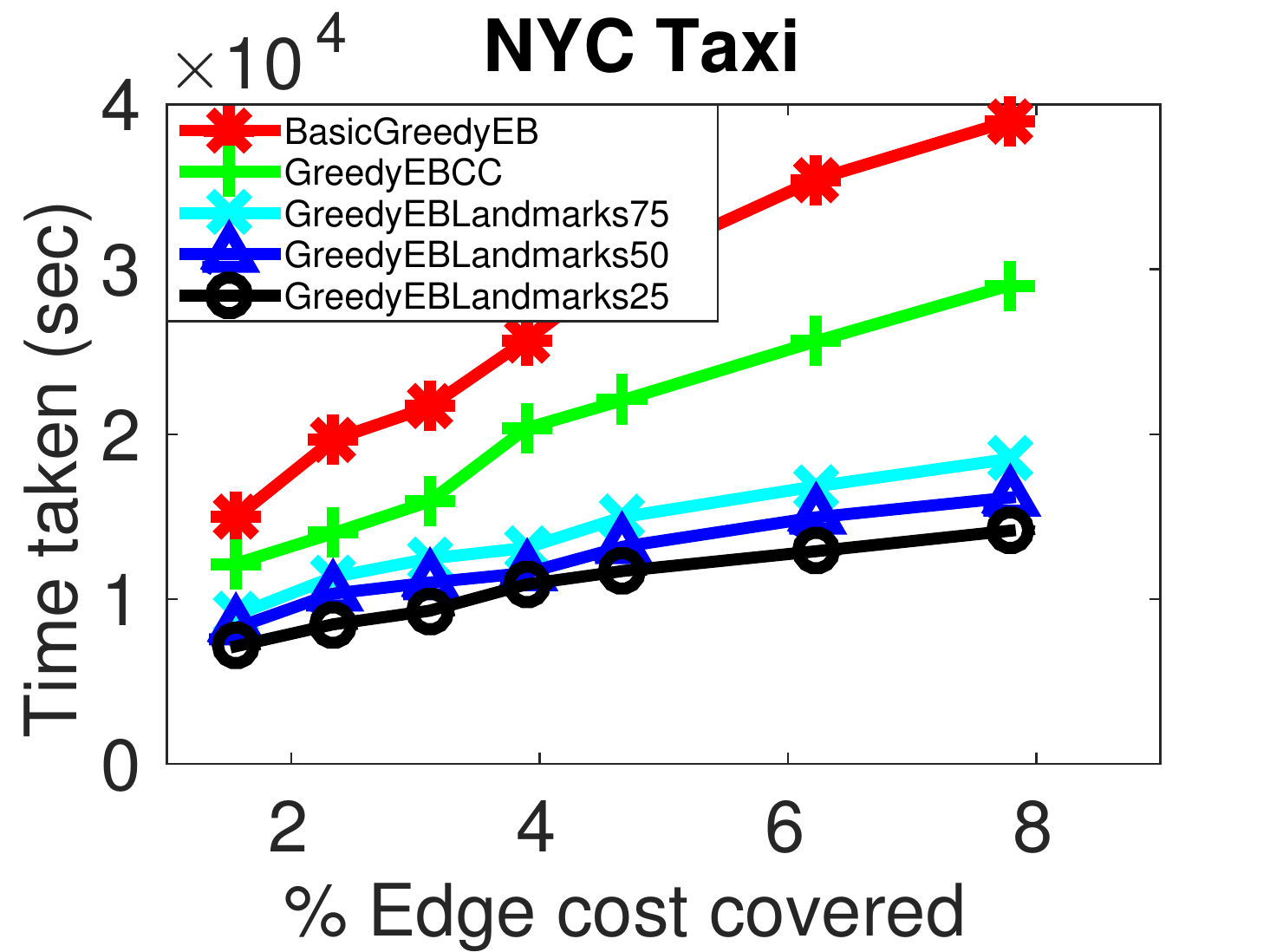}}
\end{minipage}%
\begin{minipage}{.24\linewidth}
\centering
\subfloat[]{\label{}\includegraphics[width=\textwidth, height=\textwidth]{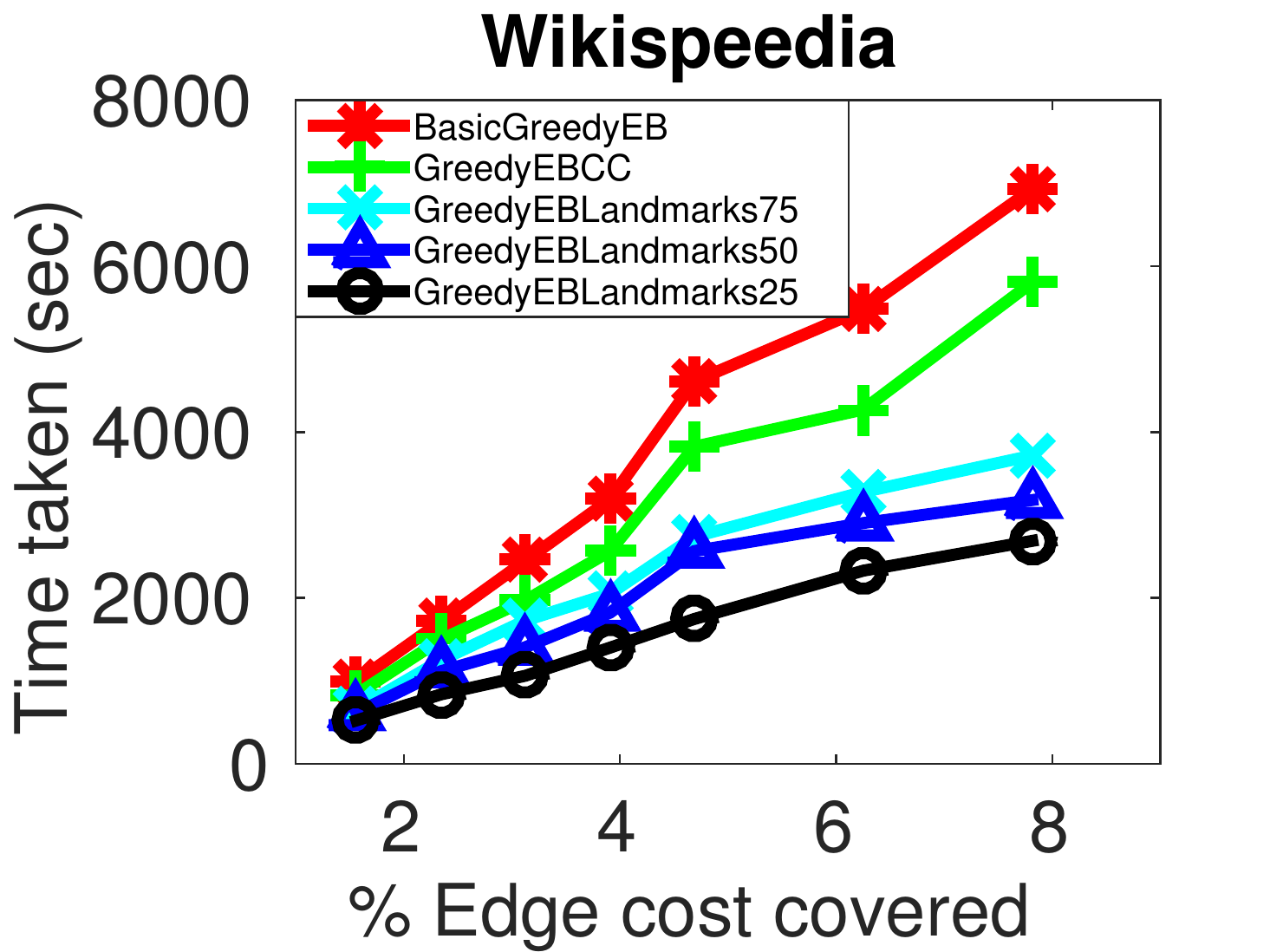}}
\end{minipage}\par\medskip

\begin{minipage}{.24\linewidth}
\centering
\subfloat[]{\label{}\includegraphics[width=\textwidth, height=\textwidth]{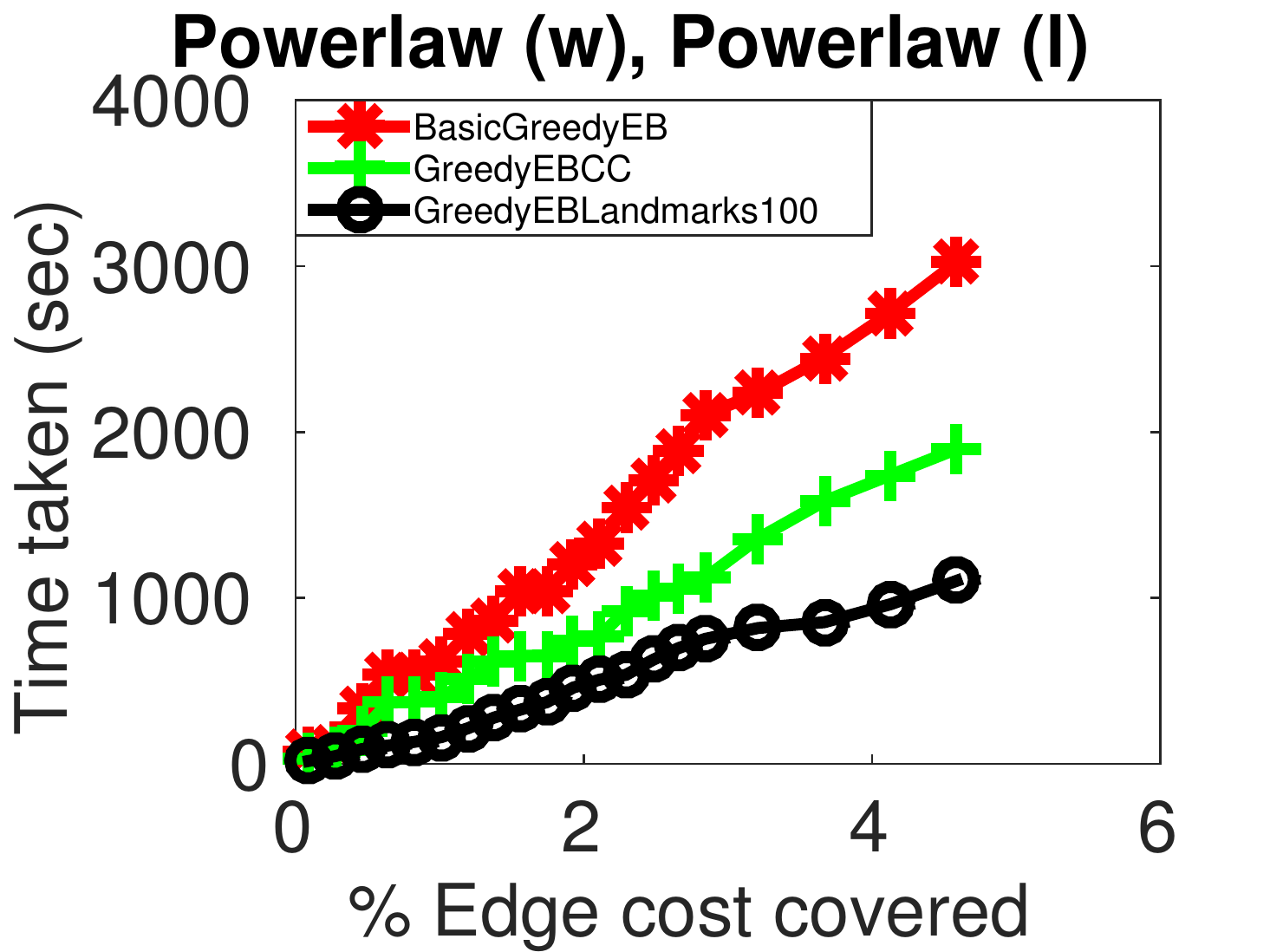}}
\end{minipage}%
\begin{minipage}{.24\linewidth}
\centering
\subfloat[]{\label{}\includegraphics[width=\textwidth, height=\textwidth]{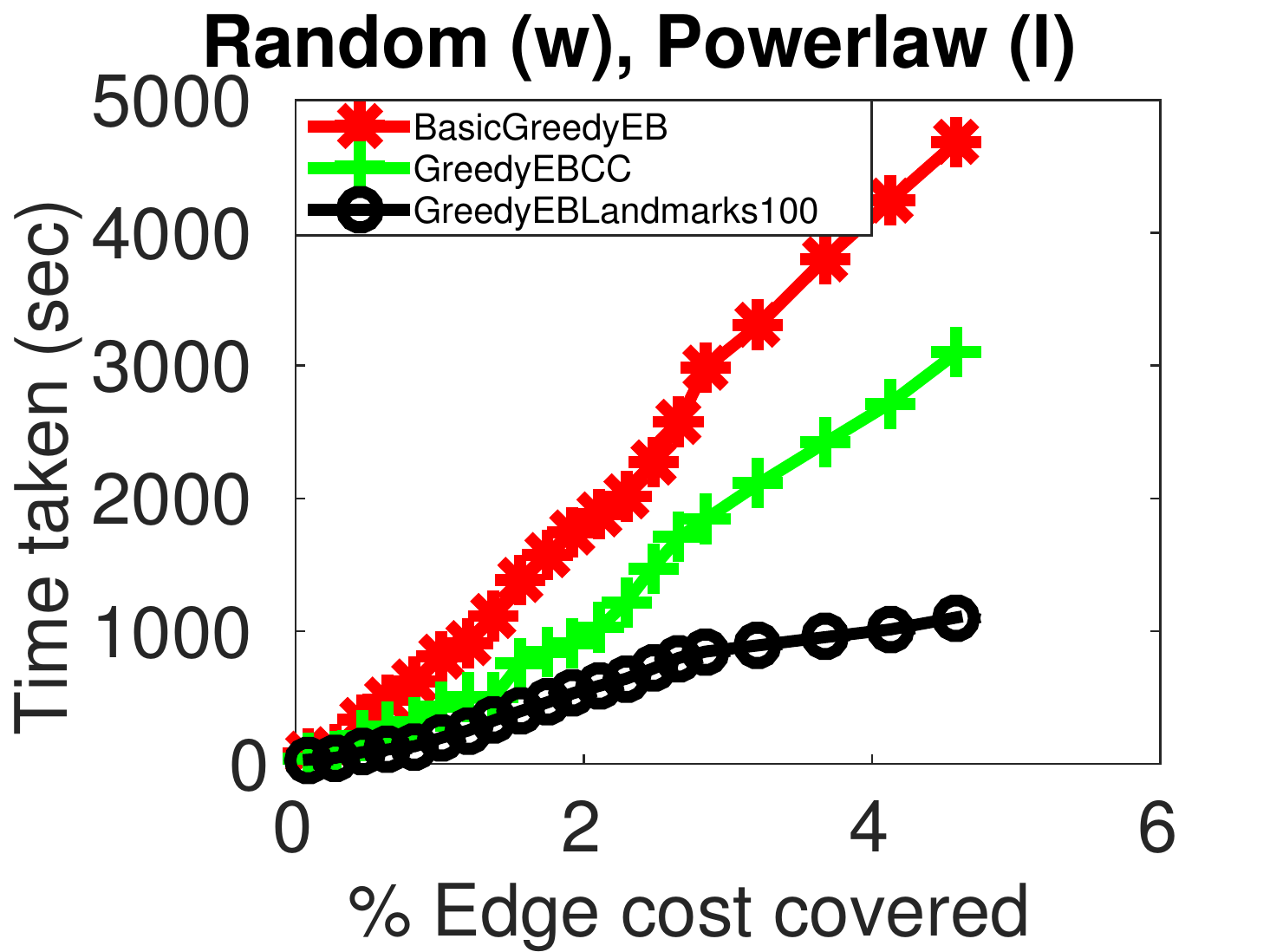}}
\end{minipage}
%
\begin{minipage}{.24\linewidth}
\centering
\subfloat[]{\label{}\includegraphics[width=\textwidth, height=\textwidth]{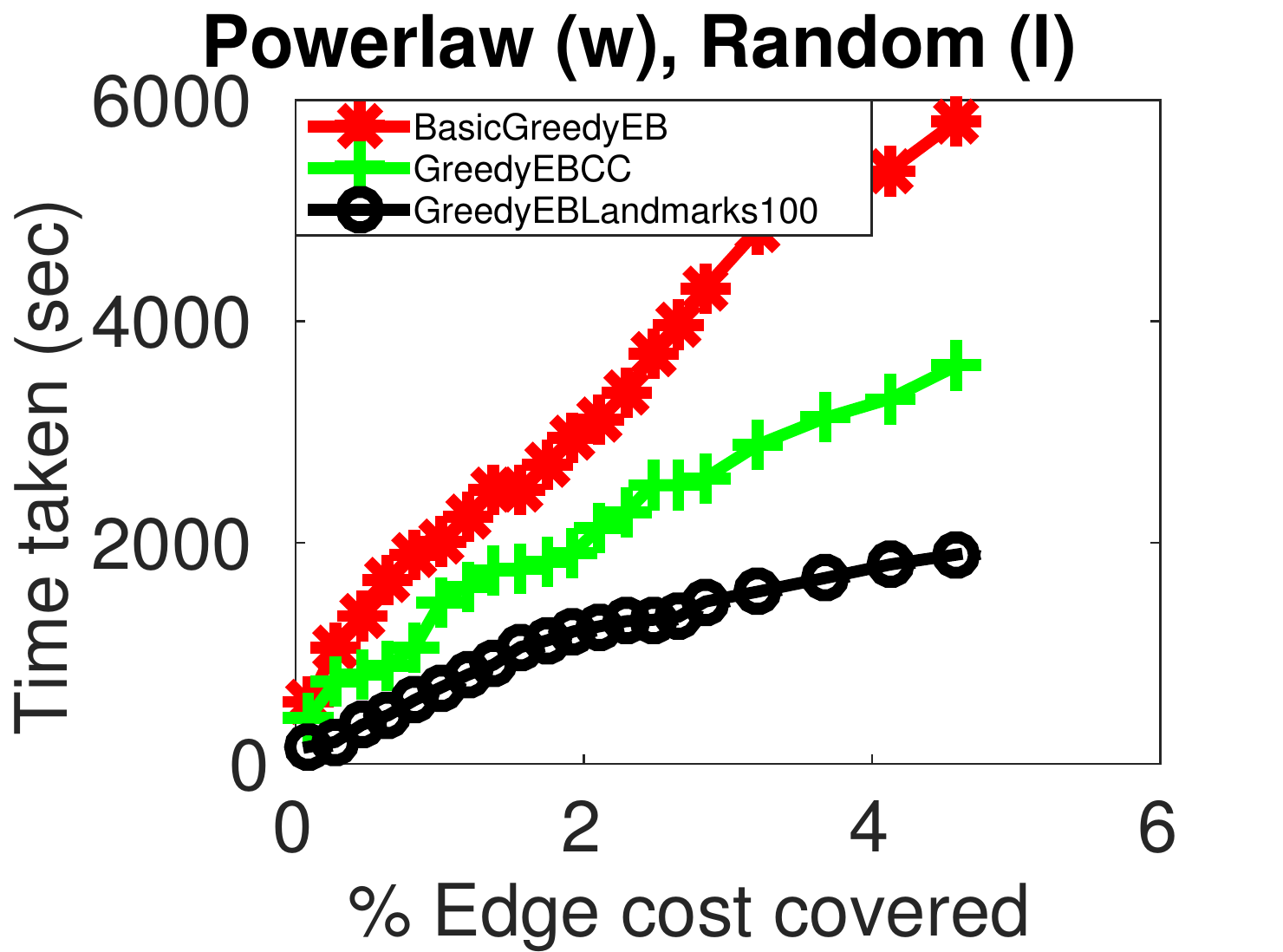}}
\end{minipage}%
\begin{minipage}{.24\linewidth}
\centering
\subfloat[]{\label{}\includegraphics[width=\textwidth, height=\textwidth]{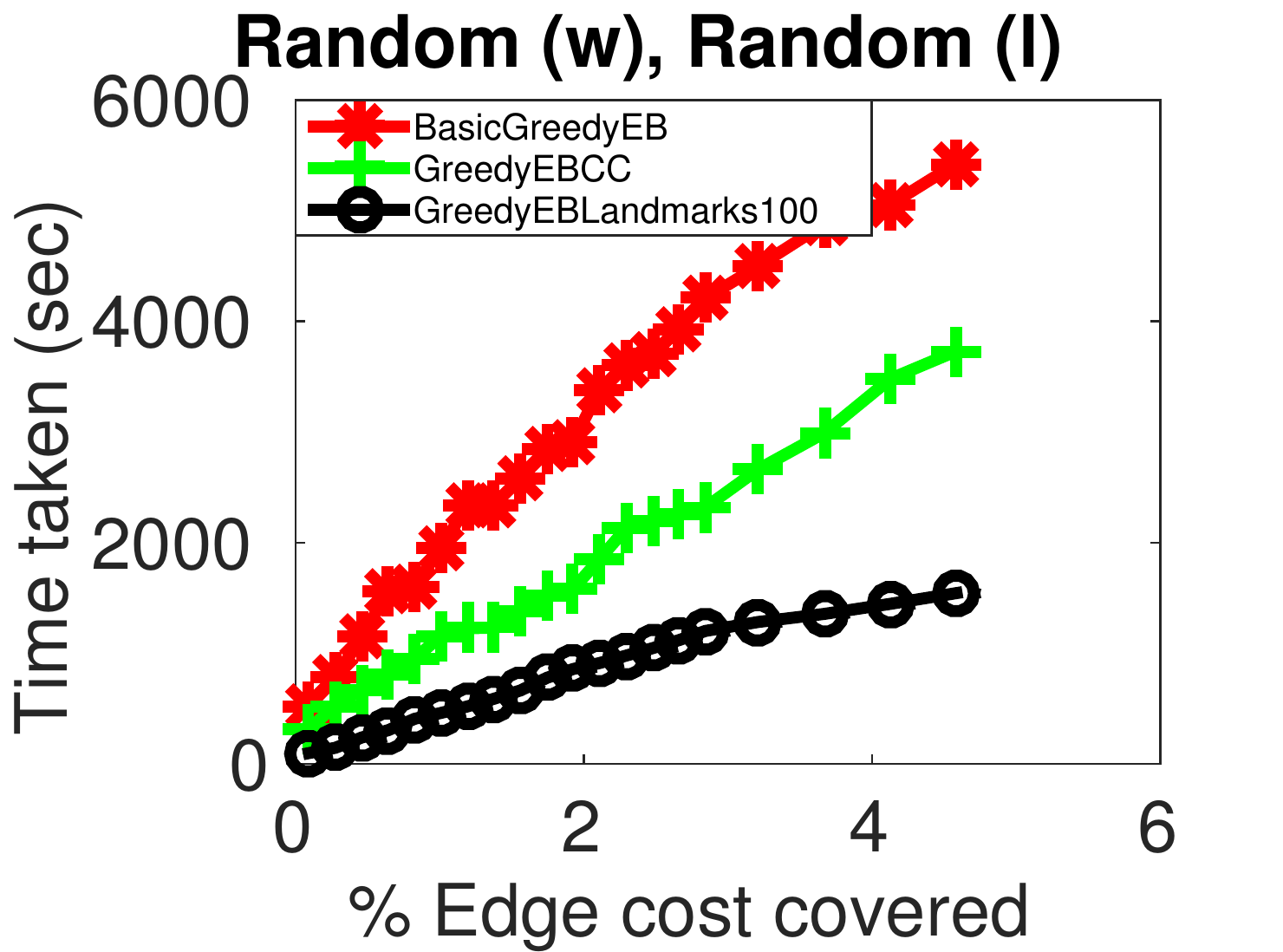}}
\end{minipage}\par\medskip

\caption{Comparison of the time taken by the algorithm using different optimizations mentioned in Section~\ref{sec:optimizations}, for (a) \london, (b) \flights, (c) \nyctaxi, (d) \wikispeedia, (e--h) \ukroad. BasicGreedyEB doesnt use any optimizations, GreedyEBCC is the version using connected components, GreedyEBLandmarks* uses * landmarks. We can clearly see a great improvement (up to 4x) in speed by using landmarks. As we increase the number of landmarks, we trade-off speed with accuracy. In Figures (e--h), $(w)$ indicates traffic volume, and $(l)$ indicates the log.}
\label{fig:landmarks_timetaken}
\end{figure*}

\begin{figure*}
\begin{minipage}{.24\linewidth}
\centering
\subfloat[]{\label{}\includegraphics[width=\textwidth, height=\textwidth]{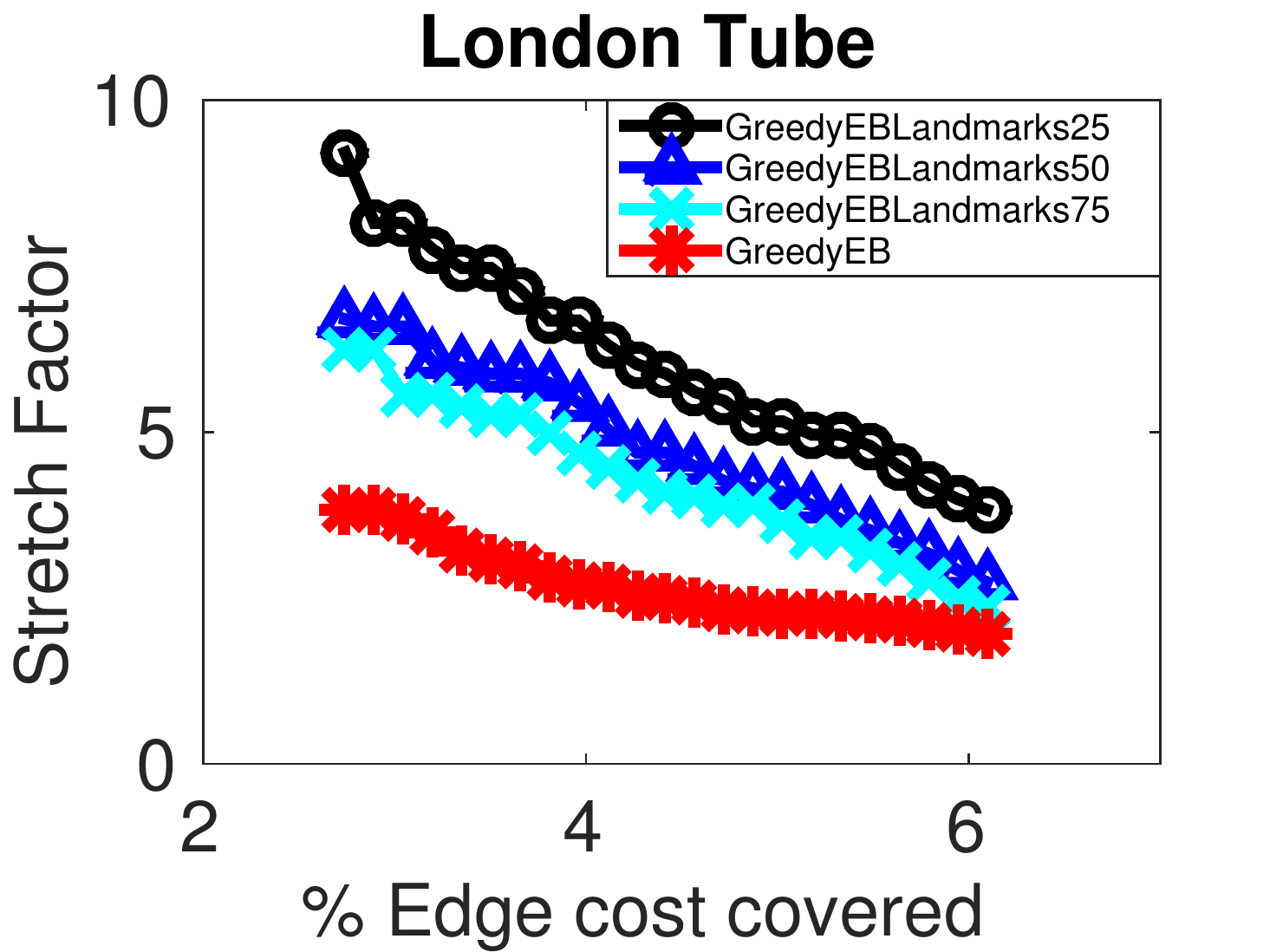}}
\end{minipage}%
\begin{minipage}{.24\linewidth}
\centering
\subfloat[]{\label{}\includegraphics[width=\textwidth, height=\textwidth]{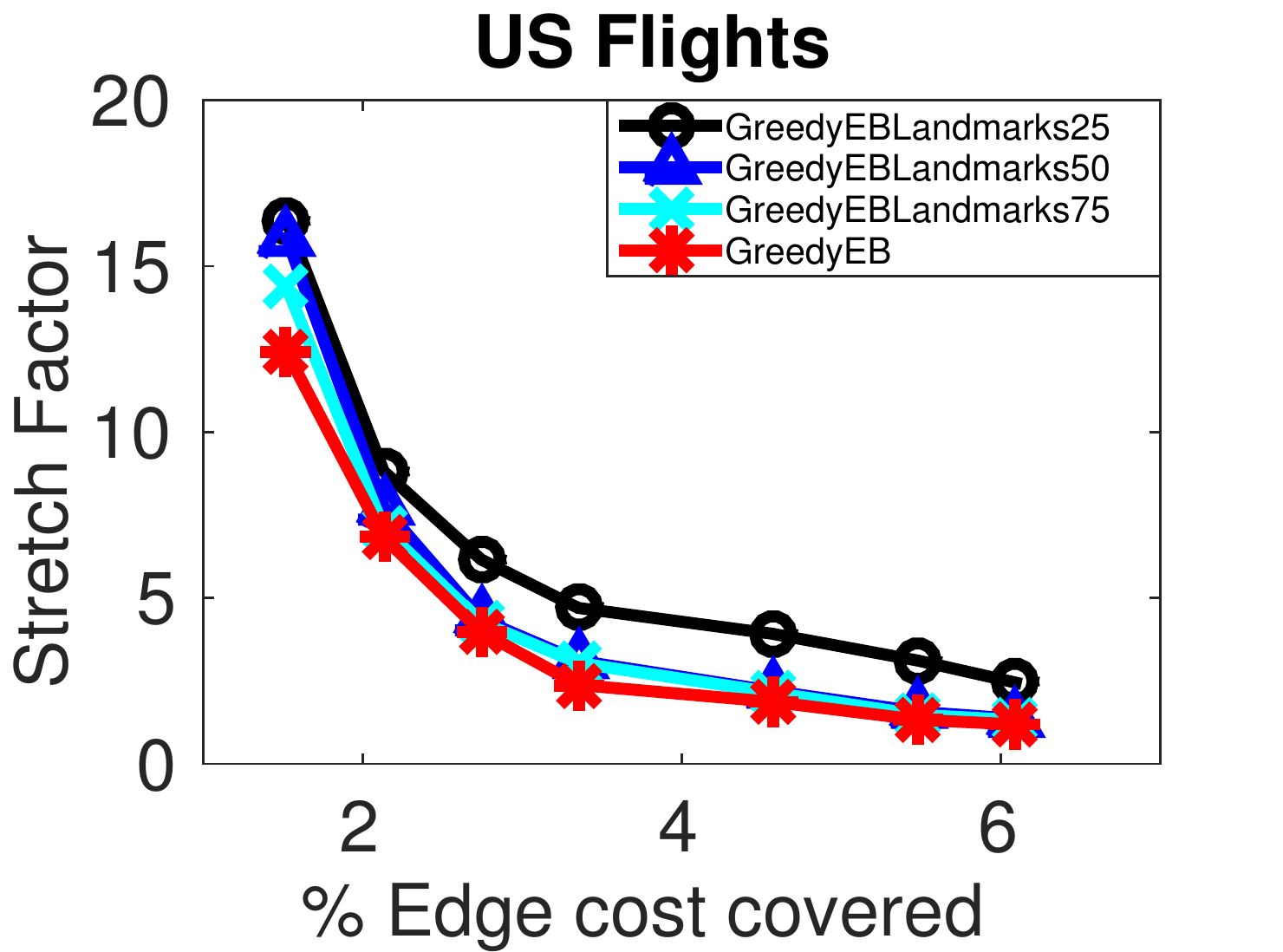}}
\end{minipage}
%
\begin{minipage}{.24\linewidth}
\centering
\subfloat[]{\label{}\includegraphics[width=\textwidth, height=\textwidth]{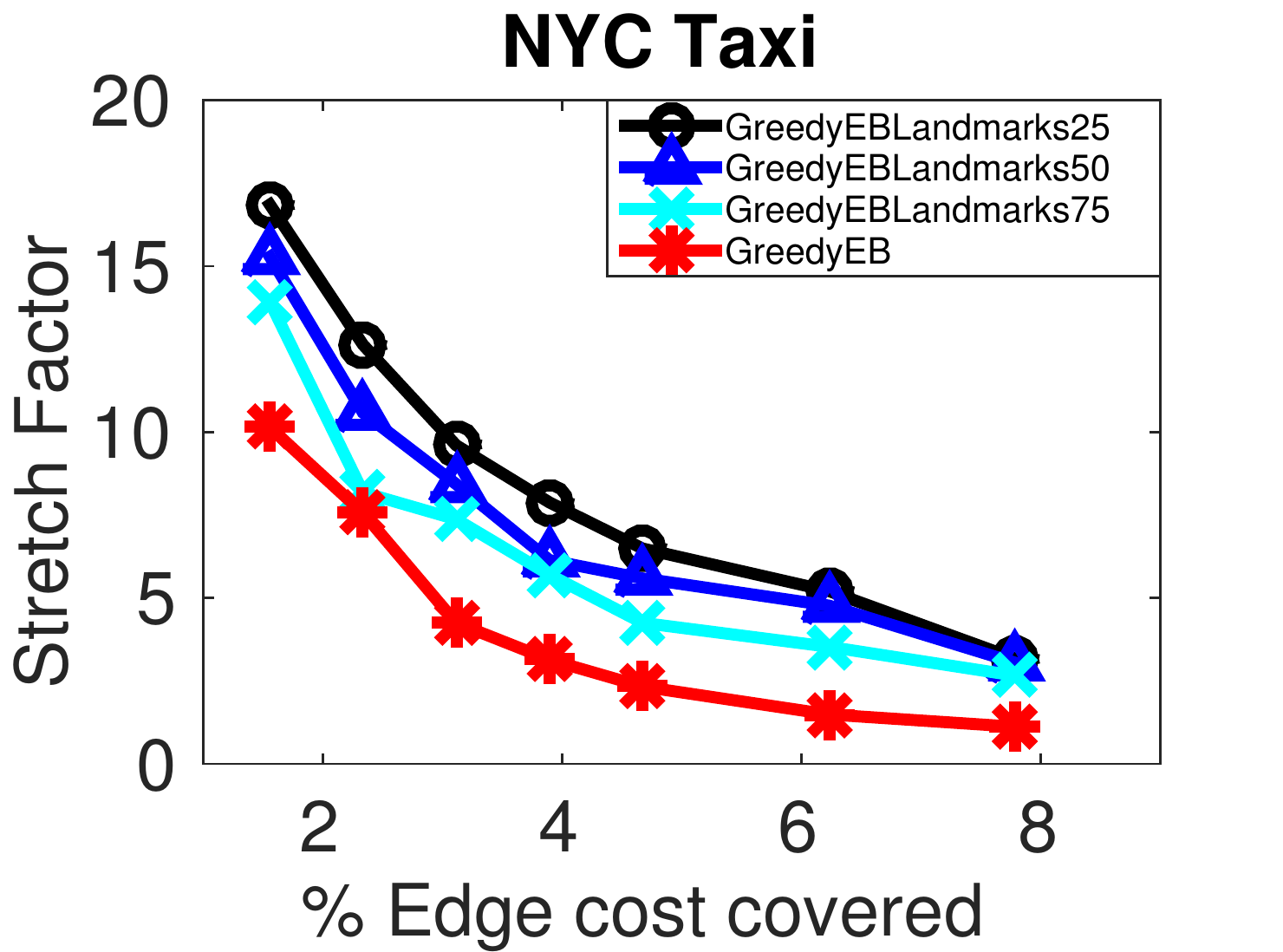}}
\end{minipage}%
\begin{minipage}{.24\linewidth}
\centering
\subfloat[]{\label{}\includegraphics[width=\textwidth, height=\textwidth]{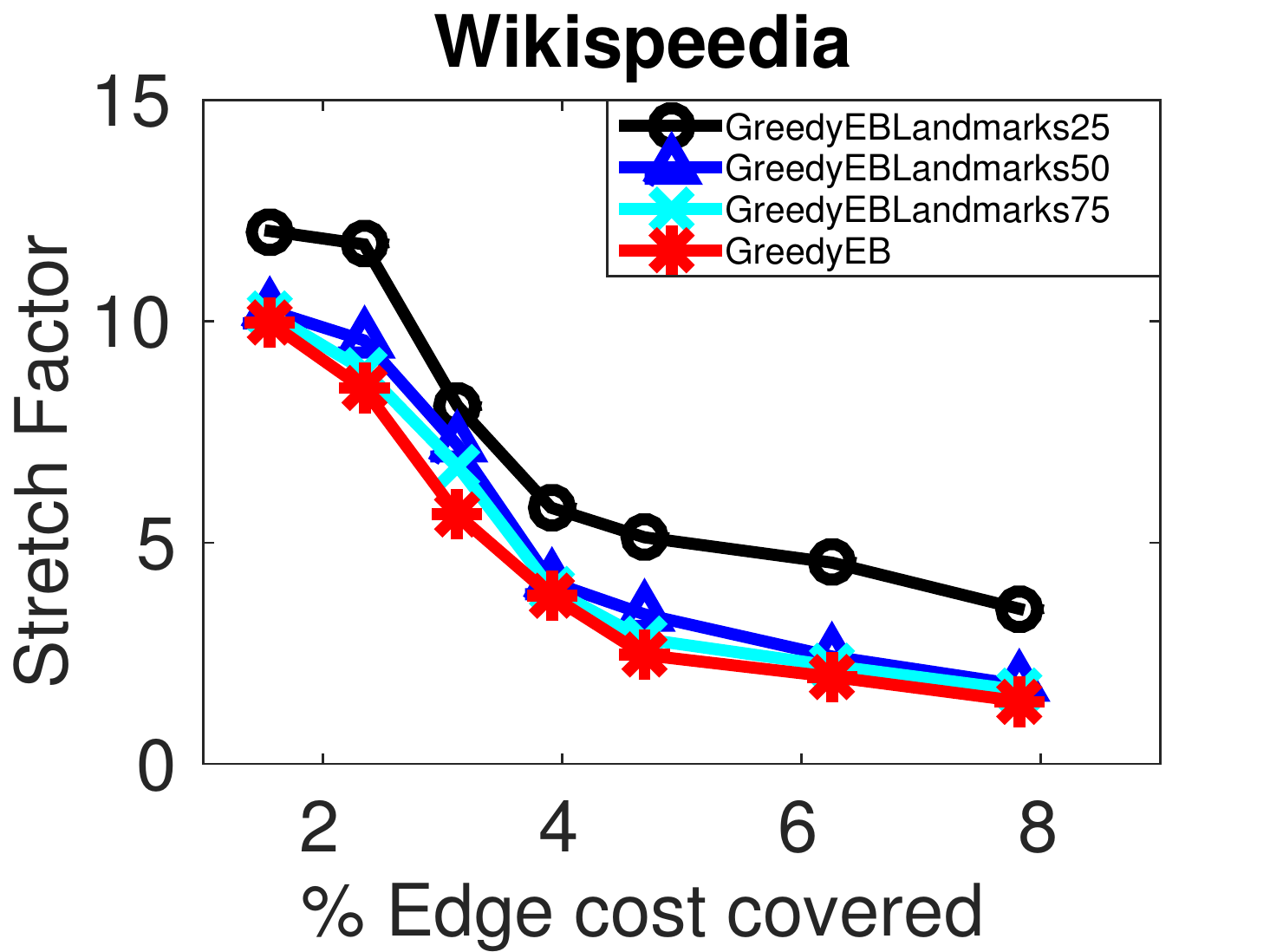}}
\end{minipage}\par\medskip

\begin{minipage}{.24\linewidth}
\centering
\subfloat[]{\label{}\includegraphics[width=\textwidth, height=\textwidth]{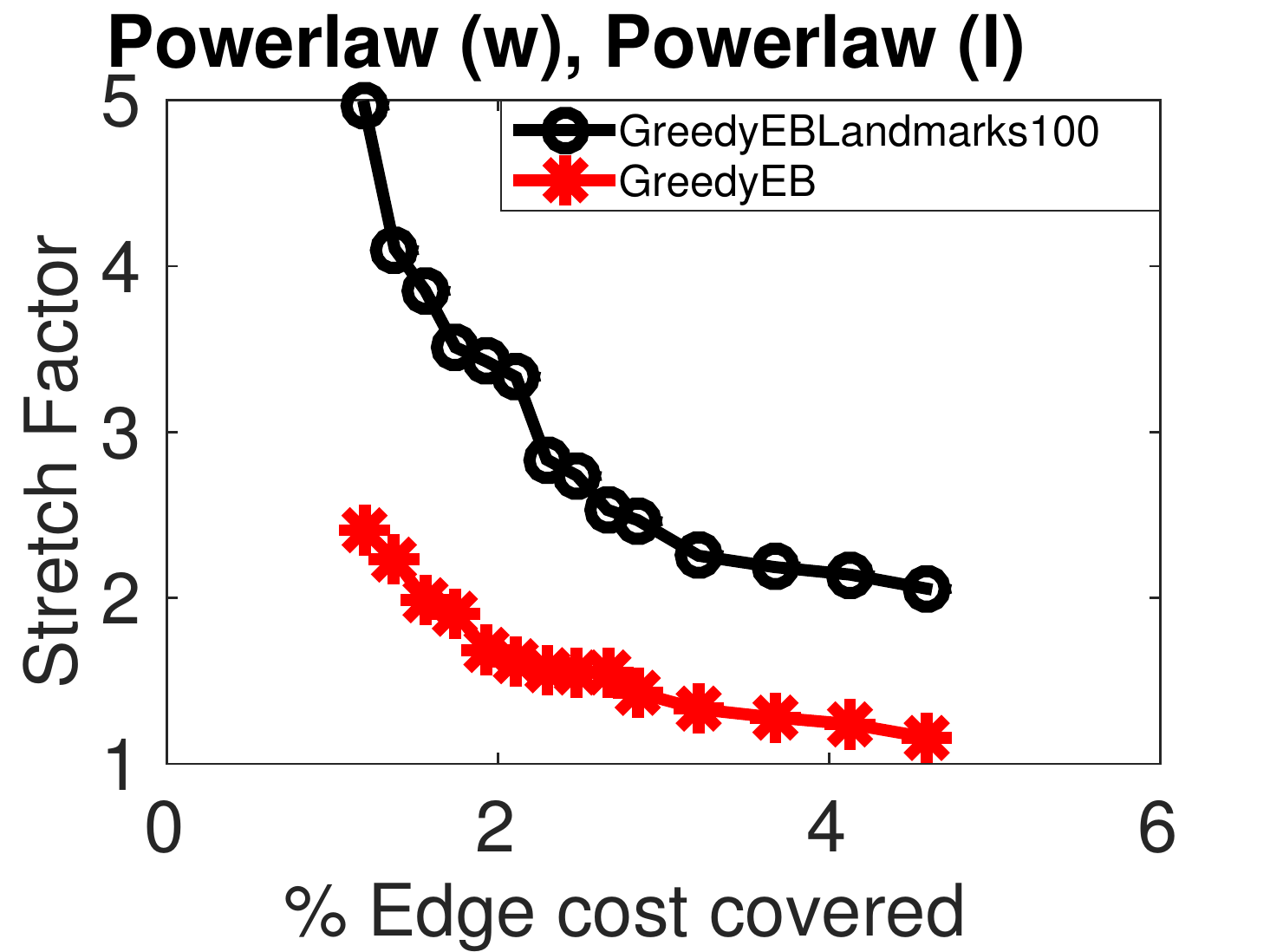}}
\end{minipage}%
\begin{minipage}{.24\linewidth}
\centering
\subfloat[]{\label{}\includegraphics[width=\textwidth, height=\textwidth]{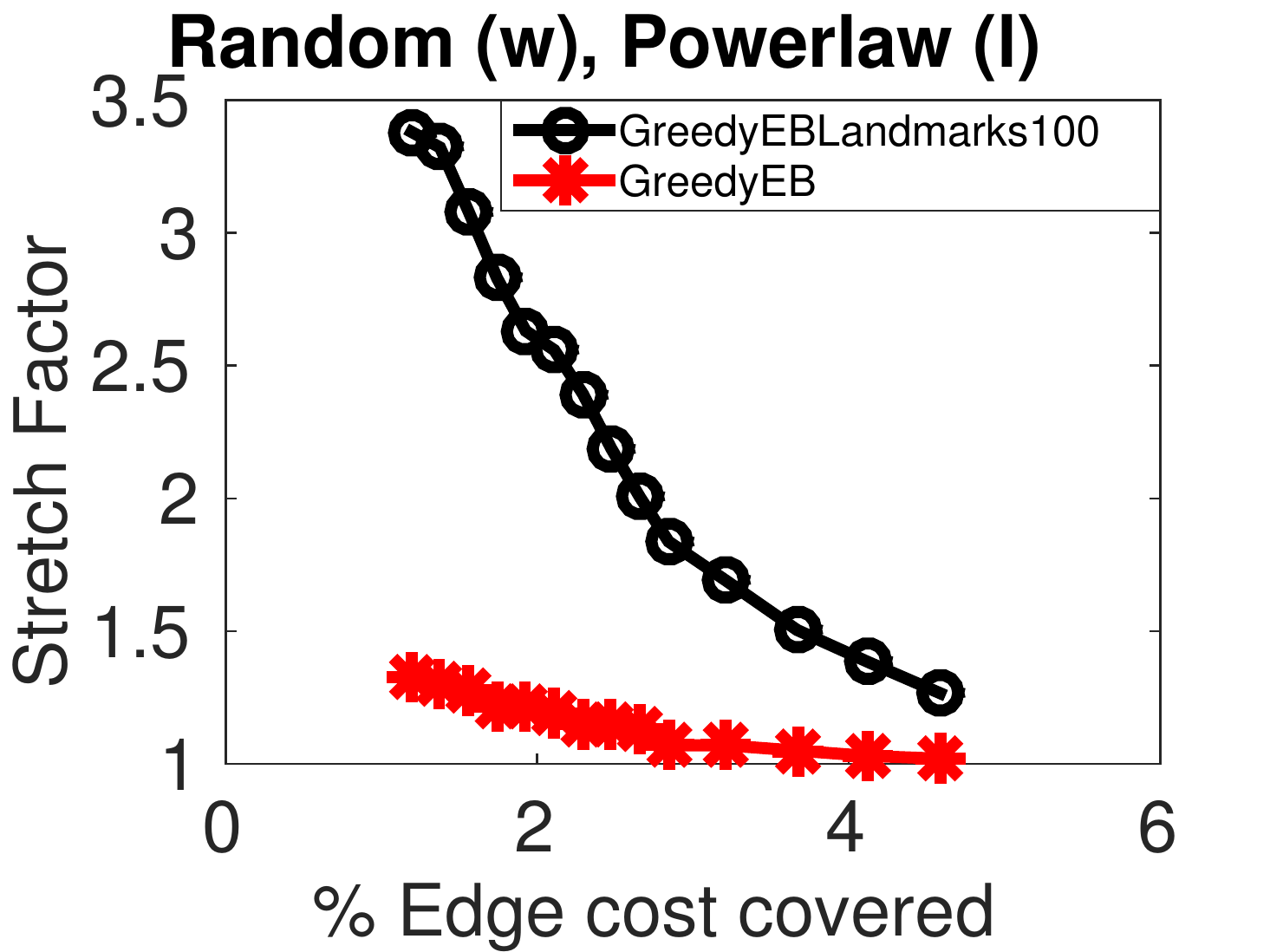}}
\end{minipage}
%
\begin{minipage}{.24\linewidth}
\centering
\subfloat[]{\label{}\includegraphics[width=\textwidth, height=\textwidth]{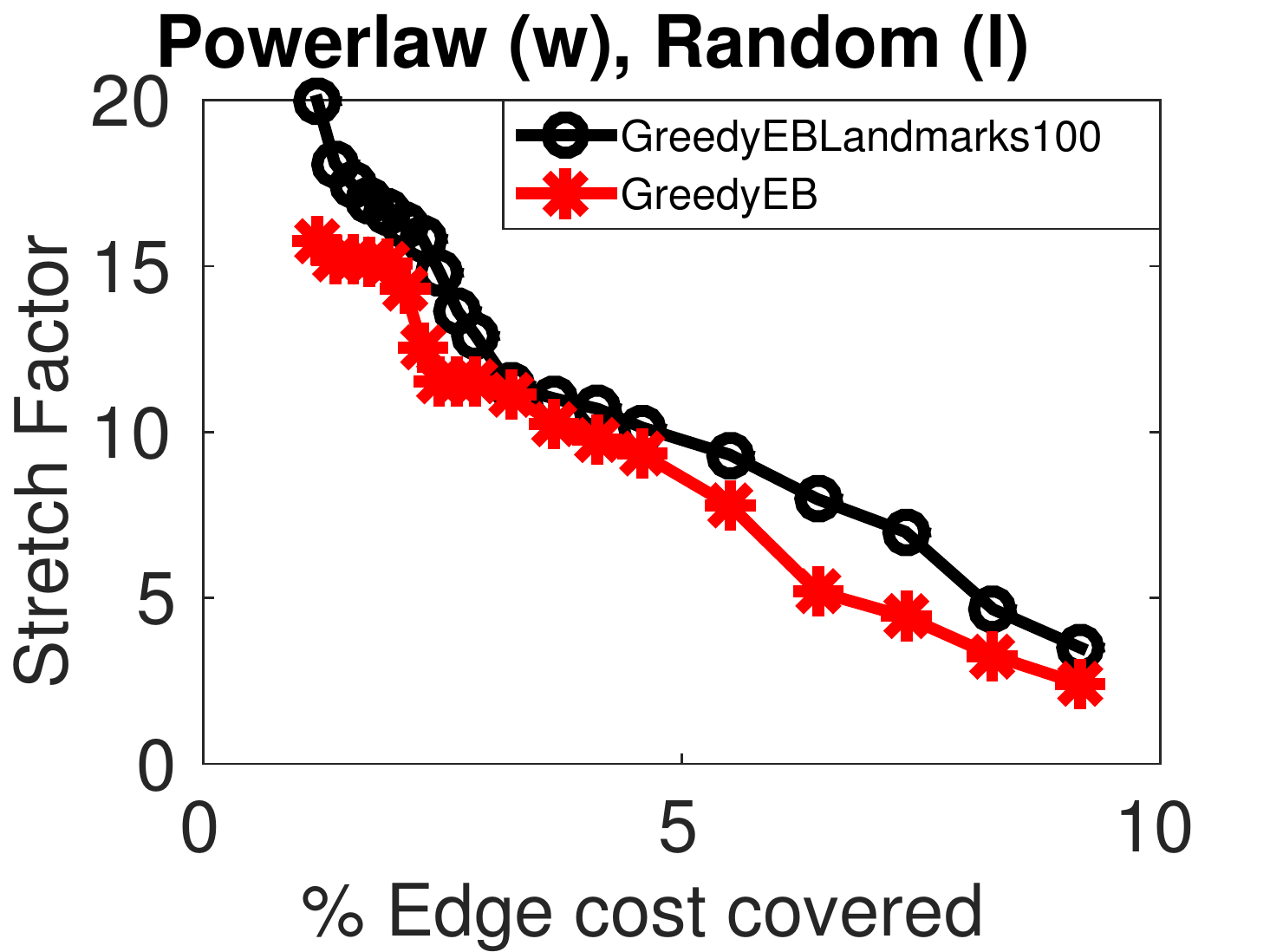}}
\end{minipage}%
\begin{minipage}{.24\linewidth}
\centering
\subfloat[]{\label{}\includegraphics[width=\textwidth, height=\textwidth]{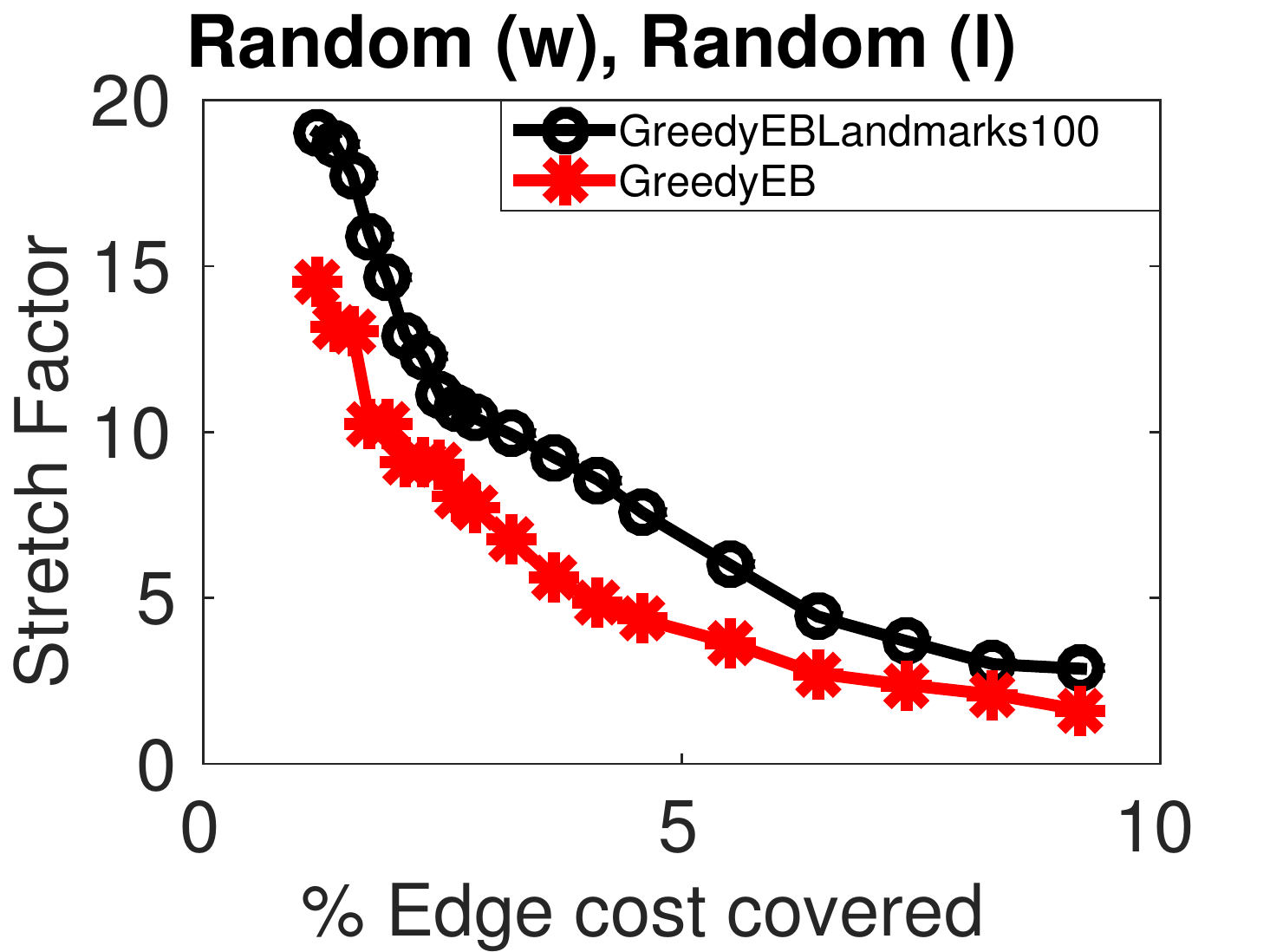}}
\end{minipage}\par\medskip

\caption{Performance in terms of stretch factor of our greedy algorithm with and with out using landmarks, for (a) \london, (b) \flights, (c) \nyctaxi\ and (d) \wikispeedia\, (e--h) \ukroad. For all the datasets, as expected, we see a slight decrease in performance using landmarks. In Figures (e--h), $(w)$ indicates traffic volume, and $(l)$ indicates the log.}
\label{fig:landmarks_performance}
\end{figure*}

\spara{Key findings}. From all the above results, we would like to highlight the following points.

\spara{1.} The greedy algorithm and its variants performs
 much better than the baseline (See Figure~\ref{fig:ebVSnormal}).
 Note that baseline is not included in Figure~\ref{fig:ebVSnormal}(g,h) because the edges in the baseline
 are added one-by-one and for a large interval of the cost, the
 stretch factor was very large or even infinity.
 This shows that the backbone produced by our greedy approach not only consists of edges with low benefit, but also tries to re-use a lot of edges,
 hence obtaining a lower stretch factor.

\spara{2.}
The backbones discovered by our algorithms are sparse and summarize
well the given traffic (Figures~\ref{fig:ebVSnormal},~\ref{fig:landmarks_performance}). 
In all cases, with about 15\% of the edge cost in the network it is
possible to summarize the traffic with stretch factor close to 1.
In some cases, even smaller budget (than 15\%) is sufficient to reach a lower
stretch-factor value.

\spara{3.}
 Incorporating edge-betweenness as an edge-weighting scheme in the
 algorithm improves the performance, in certain cases there is an
 improvement of at least 50\% (See Figure~\ref{fig:ebVSnormal}; in
 most cases, even though there is a significant improvement, the plot
 is overshadowed by a worse performing baseline). This is because,
 using edges of high centrality will make sure that these edges are
 included in many shortest paths, leading to re-using many edges.

\spara{4.} The optimizations we propose in Section~\ref{sec:optimizations} 
help in reducing the running time of our algorithm (See Figure~\ref{fig:landmarks_timetaken}).
For the optimizations not using landmarks, we see around 30\% improvement in running time.
Using landmarks substantially decreases the time taken by
the algorithms (3--4 times). While there is a compromise in the quality of the solution, we can observe from 
Figures~\ref{fig:landmarks_performance} that the performance drop is small in most cases and can be
controlled by the choosing the number of landmarks accordingly.
Our algorithms, using the various optimizations we propose, are able to scale for large, real-world networks with tens of thousands of nodes which is the typical size of a road/traffic network.

\subsection{Comparison to existing approaches}
\label{sec:comparison}

In this section, we compare the performance of \backbonediscovery\ with other related work in literature. The comparison is done based on two factors (i) Stretch factor, (ii) Percentage of edges covered by the solution. Intuitively, a good backbone should try to minimize both, i.e. produce a sparse backbone, which also preserves the shortest paths between vertices as well as possible.

\spara{Comparison with Prize Collecting Steiner-forest (PCSF)} - Prize Collecting Steiner-forest~\cite{hajiaghayi2010prize} is a variant of the classic Steiner Forest problem, which allows for disconnected source--destination pairs, by paying a penalty. The goal is to minimize the total cost of the solution by `buying' a set of edges (to connect the $s$--$t$ pairs) and paying the penalty for those pairs which are not connected. We compare the performance of our algorithm with PCSF, based on two factors (i) Stretch factor (Figure~\ref{fig:pcsf_comparison}a), (ii) Percentage of edges covered by the solution (Figure~\ref{fig:pcsf_comparison}b). We use the same ($s$,$t$) pairs that we use in our algorithm and set the traffic volume $w_i$ as the penalty score in PCSF. We first run PCSF on our datasets and compute the budget of the solution produced. Using the budget as input to our algorithm (\greedyeb), we compute our backbone.

\begin{figure}
\begin{minipage}{0.48\linewidth}
\centering
\subfloat[]{\label{}\includegraphics[width=\textwidth, height=\textwidth]{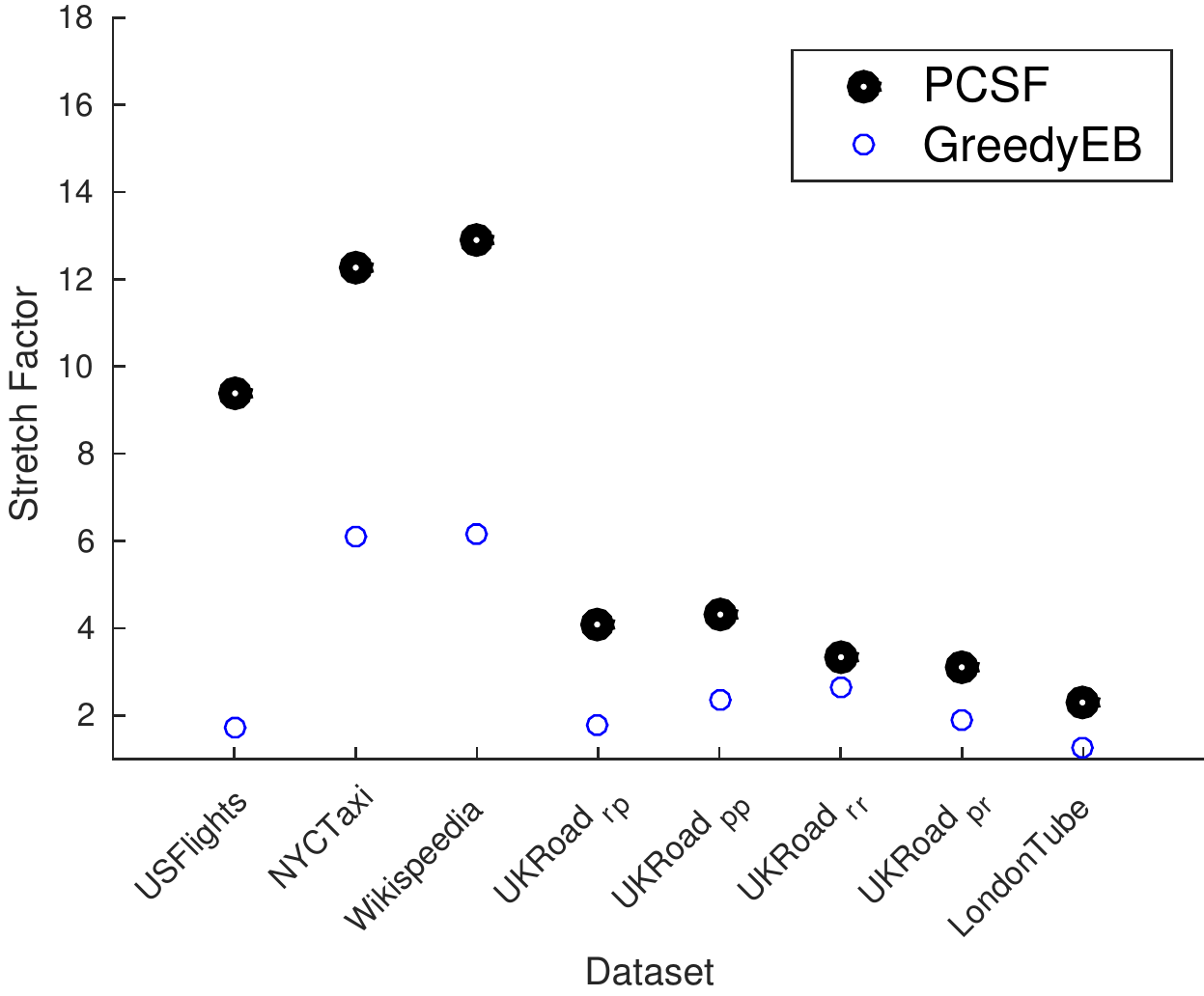}}
\end{minipage}%
\hspace{0.01\textwidth}
\begin{minipage}{0.48\linewidth}
\centering
\subfloat[]{\label{}\includegraphics[width=\textwidth, height=\textwidth]{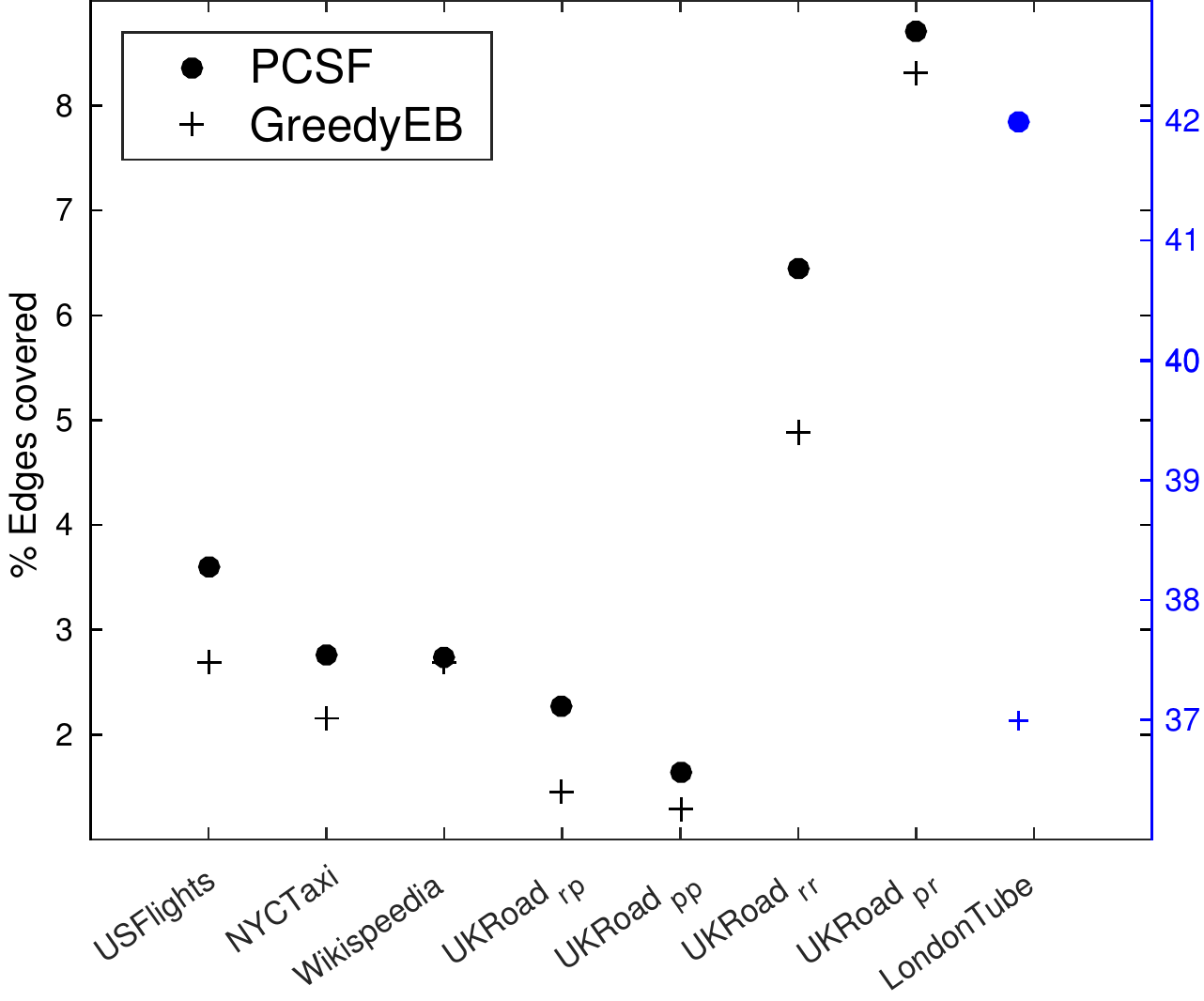}}
\end{minipage}%
\centering

\caption{Comparison of our algorithm \greedyeb\ with PCSF, in terms of (a) stretch factor (b) Percentage of edges covered. The 4 variants of \ukroad\ for the different traffic log are indicated by $UKRoad_{ab}$ where $a$ indicates traffic volume, $b$ indicates (s,t) pairs ($r$ - random, $p$ - powerlaw). (In (b) \london\ is plotted on a secondary y-axis because of mismatch in scale).}
\label{fig:pcsf_comparison}
\end{figure}

We can see from Figure~\ref{fig:pcsf_comparison}a that our algorithm produces a backbone with a much better stretch factor than PCSF. In most datasets, our algorithm produces a backbone which is at least 2 times better in terms of stretch factor.

Figure~\ref{fig:pcsf_comparison}b compares the fraction of edges covered by our algorithm and PCSF. We observe that the fraction of edges covered by our algorithm is lower than that of PCSF. This could be because our algorithm re-uses edges belonging to multiple paths. Figures~\ref{fig:pcsf_comparison}(a,b) show that even though our solution is much better in terms of stretch factor, we produce sparse backbones (in terms of the percentage of edges covered).


\spara{Comparison with k-spanner} - As described in Section~\ref{sec:related}, our problem is similar to $k$-spanner~\cite{narasimhan2007geometric} in the sense that we try to minimize the stretch factor. A $k$-spanner of a graph is a subgraph in which any two vertices are at most $k$ times far apart than on the original graph. One of the main advantages of our algorithm compared to spanners is that spanners can not handle disconnected vertices. We also propose and optimize a modified version of stretch factor in order to handle disconnected vertices. Similar to PCSF, we first compute a 2-spanner using a 2 approximation greedy algorithm and compute the budget used. We then run our algorithm for the same budget. Figures~\ref{fig:spanner_comparison}(a,b) show the performance of our algorithm in terms of stretch factor and percentage of edges covered. Our objective here is to compare the cost our algorithm pays in terms of stretch factor for allowing disconnected vertices. We can clearly observe that even though we allow for disconnected pairs, our algorithm performs slightly better in terms of stretch factor and also produces a significantly sparser backbone.

\begin{figure}[ht]
\begin{minipage}{.48\linewidth}
\centering
\subfloat[]{\label{}\includegraphics[width=\textwidth, height=\textwidth]{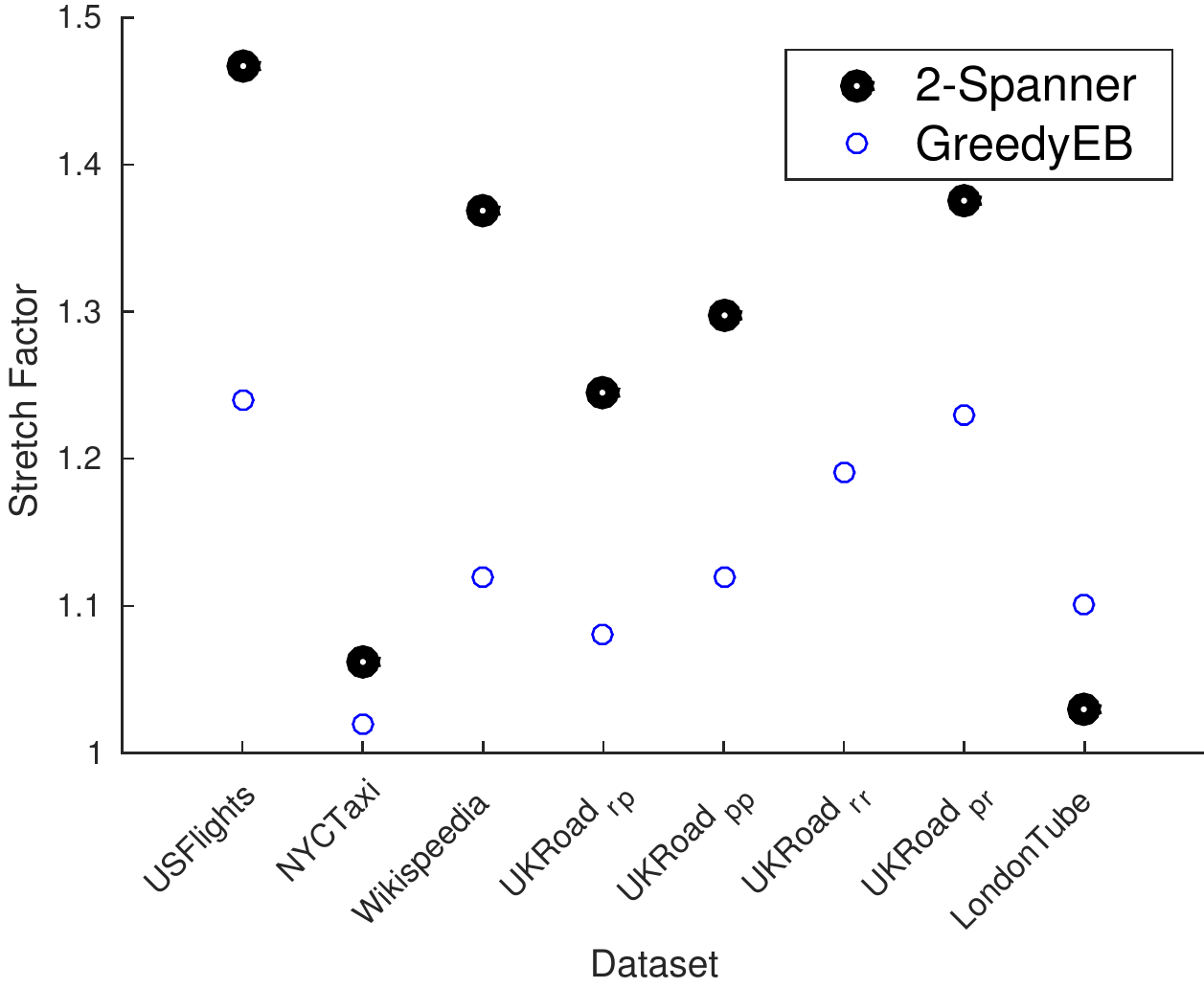}}
\end{minipage}%
\hspace{0.01\textwidth}
\begin{minipage}{.48\linewidth}
\centering
\subfloat[]{\label{}\includegraphics[width=\textwidth, height=\textwidth]{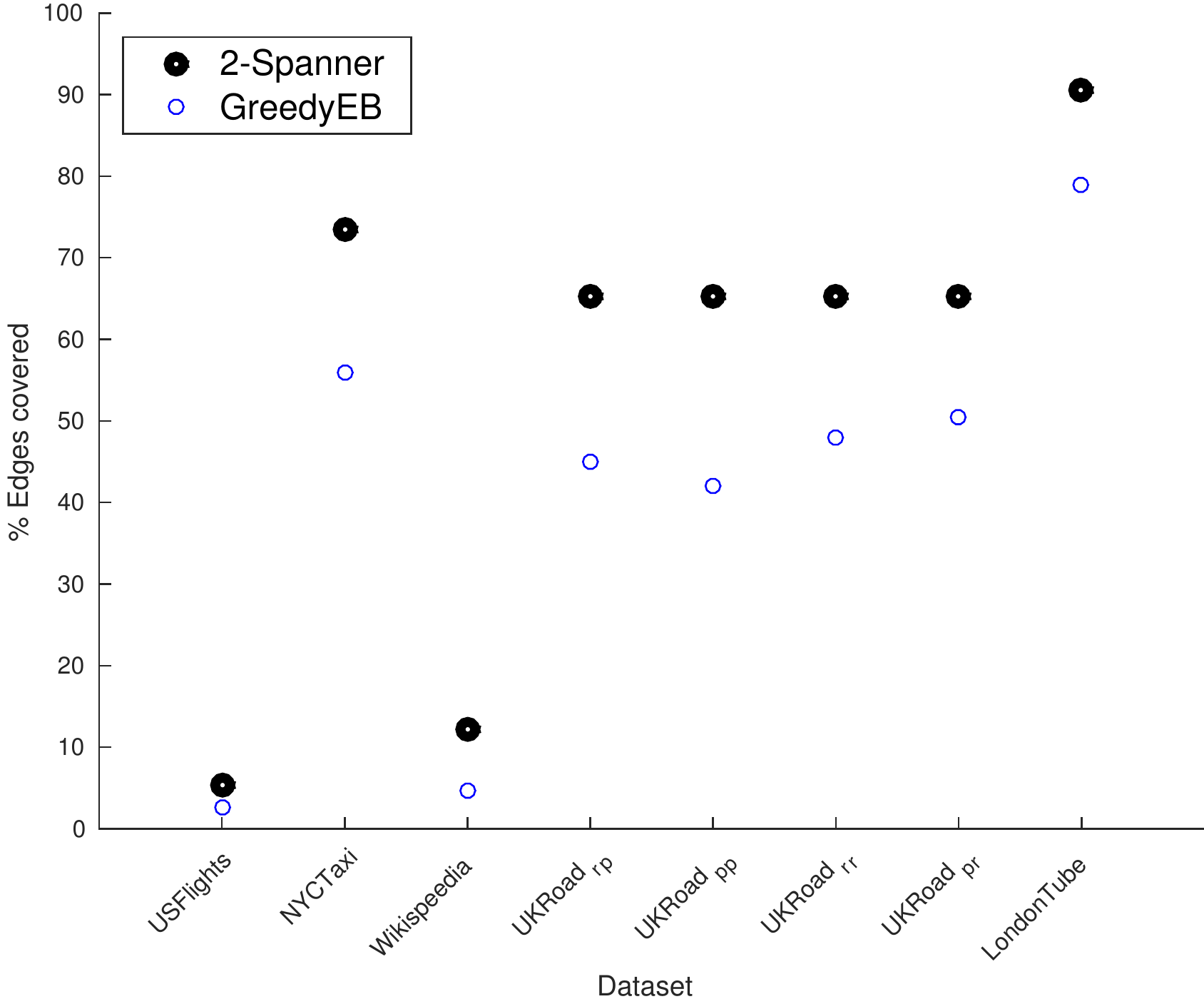}}
\end{minipage}%

\caption{Comparison of our algorithm \greedyeb\ with 2-spanner in terms of (a) stretch factor (b) Percentage of edges covered.}
\label{fig:spanner_comparison}
\end{figure}

\spara{Comparison with Toivonen et al.~\cite{ToivonenMZ10}} - Next, we compare our algorithm with Toivonen, et al~\cite{ToivonenMZ10}. Toivonen et al. propose a framework for path-oriented graph simplification, in which edges are pruned while keeping the original quality of the paths between all pairs of nodes. The objective here is to check how well we perform in terms of graph sparsification. Figures~\ref{fig:toivonen_comparison}(a,b) shows the comparison in terms of stretch factor and percentage of edges covered. Similar to the above approaches, we use the same budget as that used by Toivonen's algorithm. We observe that for most of the datasets, their algorithm works poorly in terms of sparsification, pruing less than 20\% of the edges (Figure~\ref{fig:toivonen_comparison}(b)). Our algorithm performs better both in terms of the stretch of the final solution as well as sparseness of the backbone.

\begin{figure}[ht]
\begin{minipage}{.48\linewidth}
\centering
\subfloat[]{\label{}\includegraphics[width=\textwidth, height=\textwidth]{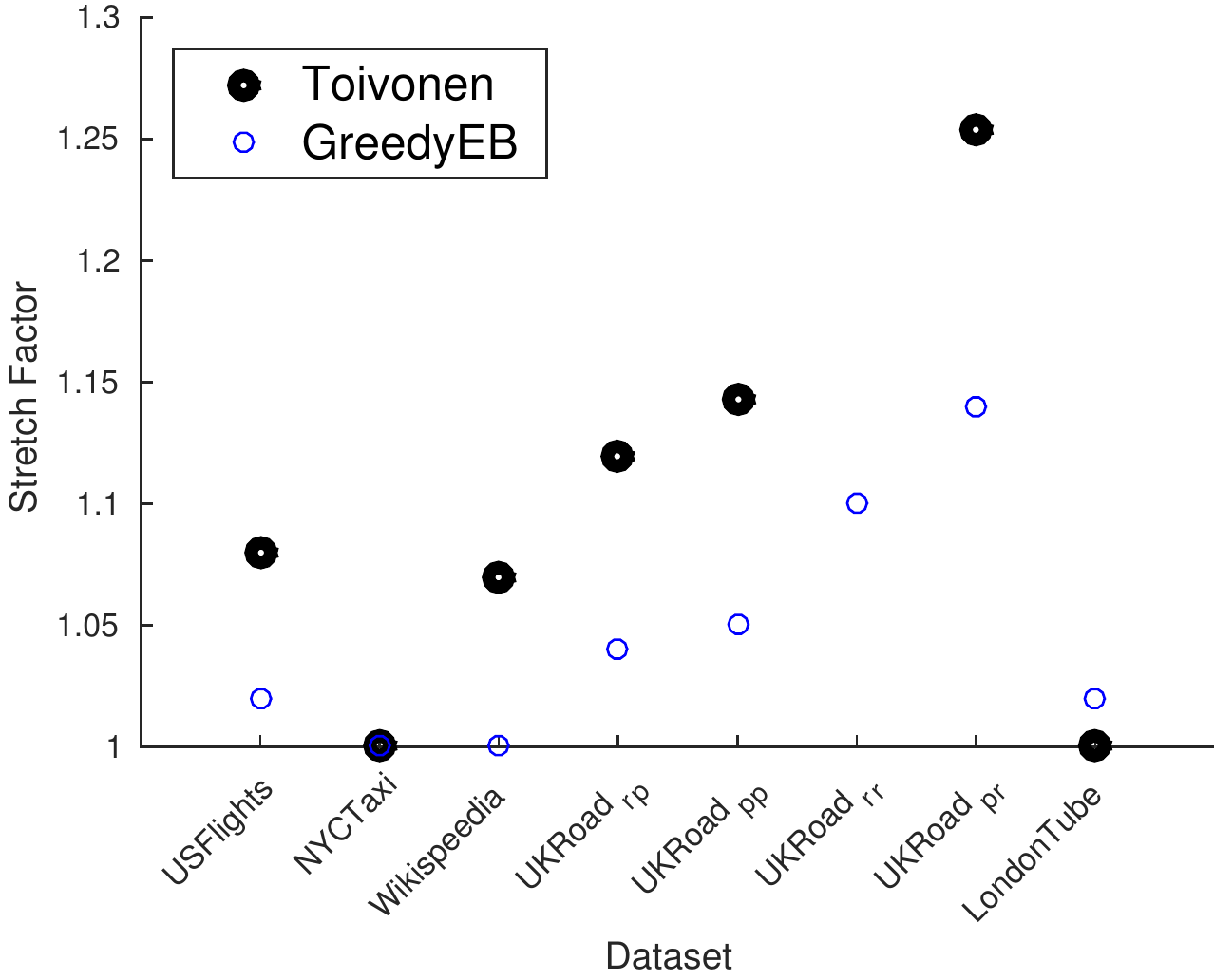}}
\end{minipage}%
\hspace{0.01\textwidth}
\begin{minipage}{.48\linewidth}
\centering
\subfloat[]{\label{}\includegraphics[width=\textwidth, height=\textwidth]{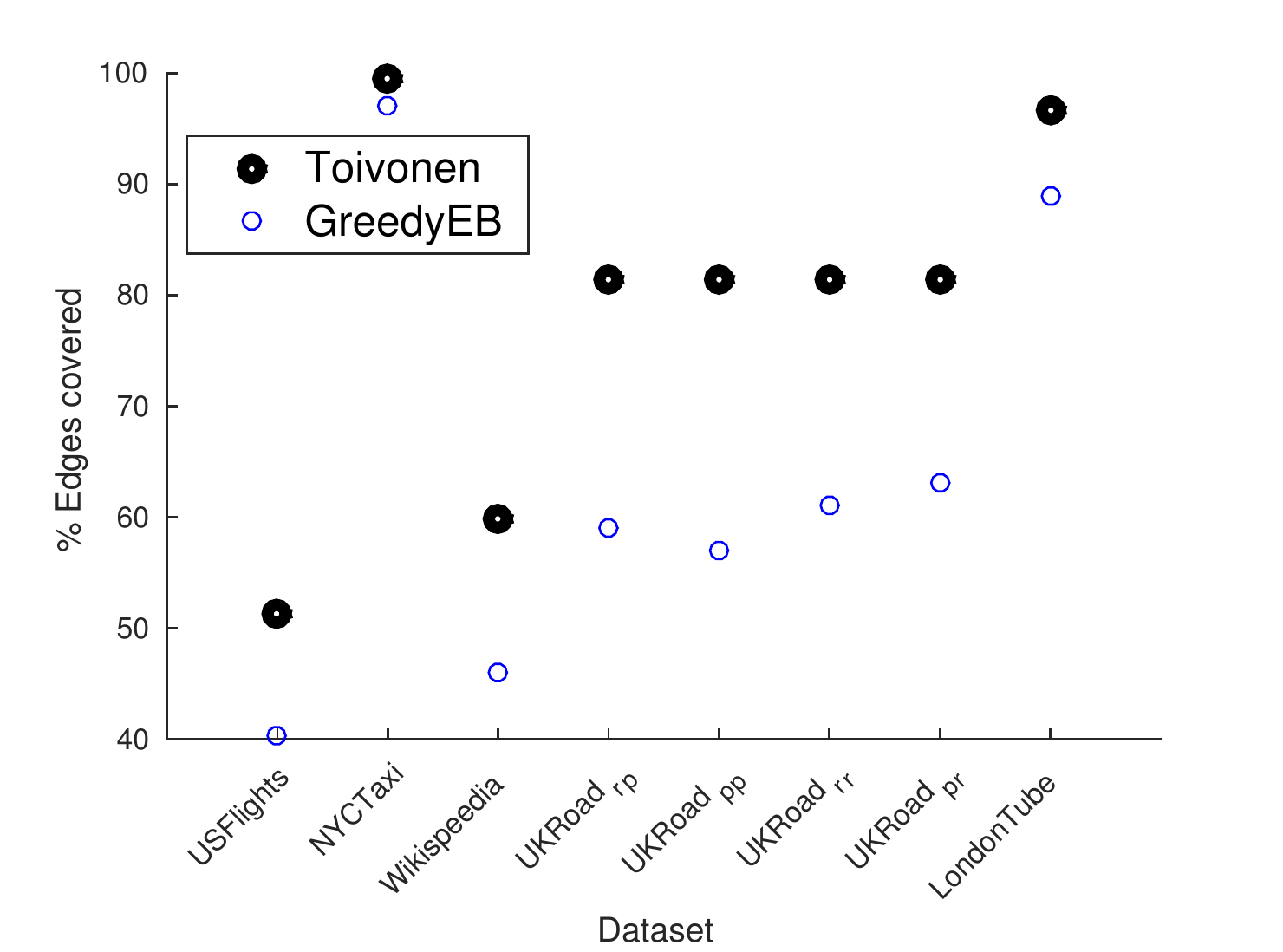}}
\end{minipage}%

\caption{Comparison of our algorithm \greedyeb\ with Toivonen et al, in terms of (a) stretch factor (b) Percentage of edges covered.}
\label{fig:toivonen_comparison}
\end{figure}

The above results, comparing our work with the existing approaches showcase the power of our algoritm in finding a concise representation of the graph, at the same time maintaining a low stretch factor. In all the three cases, our algorithm performs considerably better than the related work.

\spara{Fairness} - Though we claim that our approach performs better, we need to keep in mind that there might be differences between these algorithms. PCSF does not optimize for stretch factor. Spanners and Toivonen et al. do not have a traffic log (($s$,$t$) pairs). They also do not try to optimize stretch factor. For this section, we were just interested in contrasting the performance of our approach with existing state of the art methods and show how our approach  is different and better at what we do. 


\subsection{Case study \#1: \nyctaxi.}
The backbone of the NYC taxi traffic, as discovered by our algorithms
\greedy\ and \greedyeb, is shown in Figure~\ref{fig:nycbackbone}.
We see that both backbones consist of many street stretches in the
mid-town (around Times Square) while serving lower-town (Greenwich
village and Soho) and up-town (Morningside heights). 
We also note that there are stretches to the major transportation centers,
such as the LaGuardia airport, the World Financial Center Ferry
Terminal, and the Grand Central Terminal, as well as to the
Metropolitan museum.
Comparing the \greedy\ and \greedyeb\ backbones, we see that 
\greedyeb\ emphasizes more on the traffic to 
lower-town, and ignores the northern stretch via Robert Kennedy bridge, as 
it is less likely to be included in many shortest paths. The 
case study reiterates the advantages of using edge-betweenness to
guide the selection of the backbone to include edges which are
likely to be used more and is consistent with the well established
notion of Wardrop Equilibrium in Transportation Science that
users (in a non-cooperative manner) seek to minimize their cost of 
transportation~\cite{wardrop}. 

\begin{figure*}[t]
\begin{minipage}{.48\linewidth}
\centering
\subfloat[]{\label{}\includegraphics[width=\textwidth, height=\textwidth, clip=true, trim=5 0 5 5]{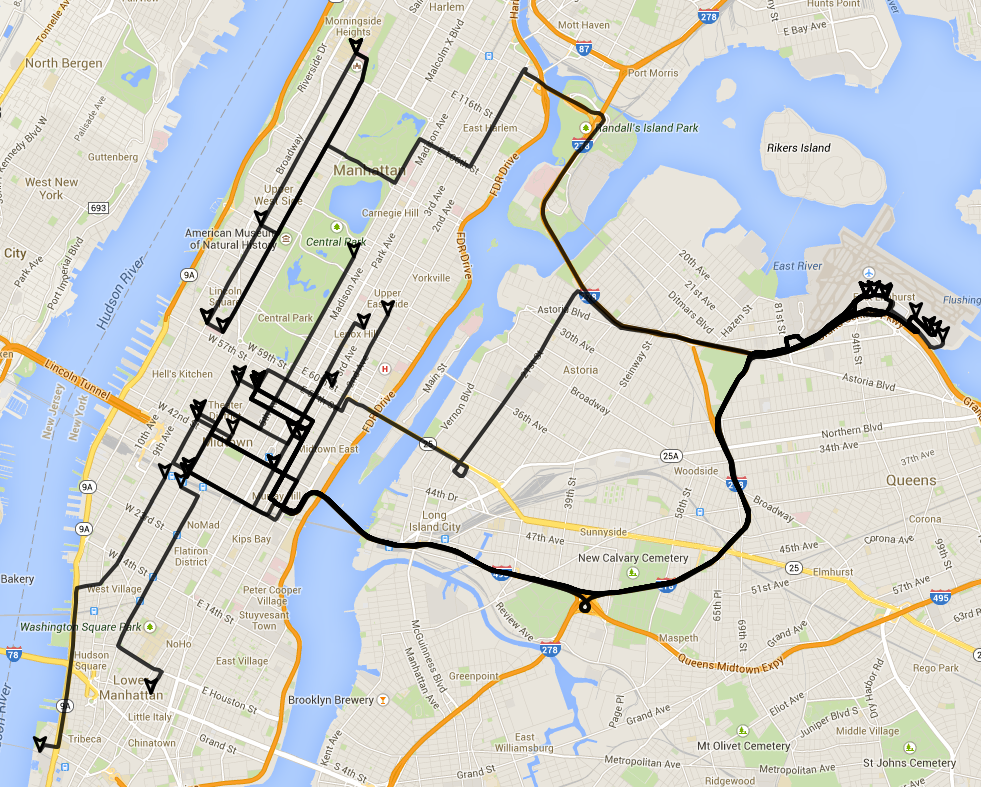}}
\end{minipage}%
\hspace{0.01\textwidth}
\begin{minipage}{.48\linewidth}
\centering
\subfloat[]{\label{}\includegraphics[width=\textwidth, height=\textwidth, clip=true, trim=5 0 5 5]{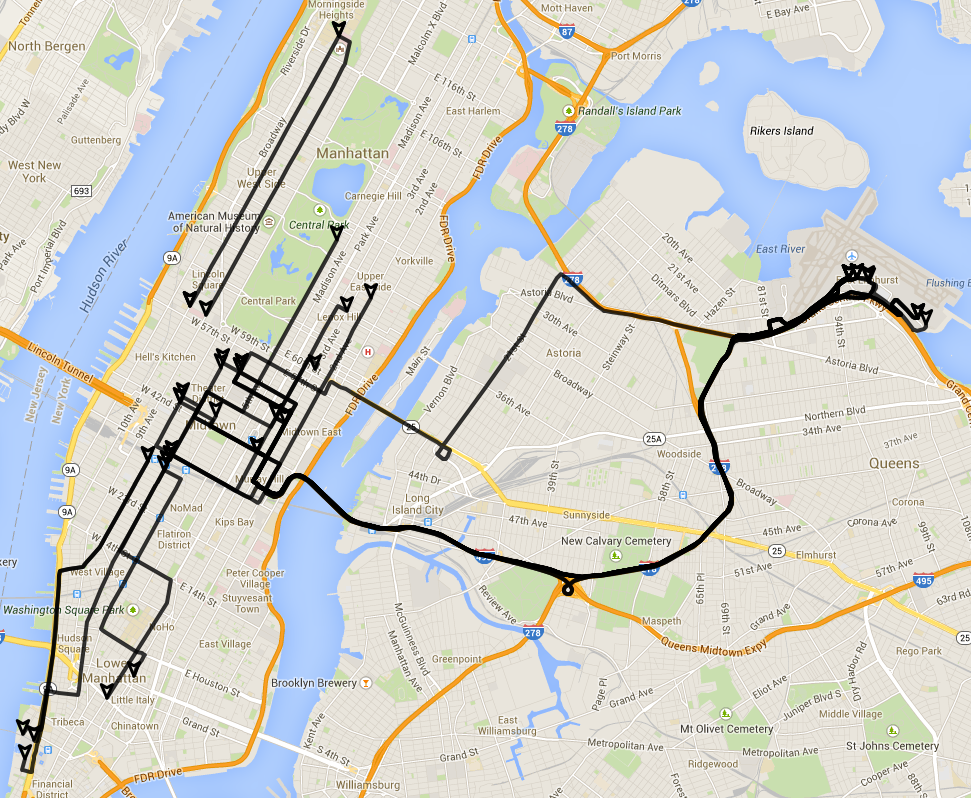}}
\end{minipage}\par\medskip

\caption{NYC backbone using (a) \greedy (b) \greedyeb.}
\label{fig:nycbackbone}
\end{figure*}

\subsection{Case study \#2: \abeline.}
\label{sec:abeline}
We carry out a qualitative analysis 
on the Abilene dataset. 
The results of applying the \greedy\ algorithm are shown
in Figure \ref{fig:abeline}.\footnote{The two nodes in Atlanta have been merged.}
The results provide preliminary evidence that the backbone produced by our
problem can be tightly integrated with software defined networks
(SDN), an increasingly important area in communication
networks~\cite{KimF13}. The objective of SDN is to allow a software
layer to control the routers and switches in the physical layers based
on the profile and shape of the traffic. 
This is precisely what our
solution is
accomplishing in Figure~\ref{fig:abeline}. 
The design of data-driven logical networks will be an important
operation implemented through an SDN and will help network designers
manage traffic in real time.

\begin{figure*}
\begin{minipage}{.48\linewidth}
\centering
\subfloat[]{\label{}
	\includegraphics[width=\textwidth, height=\textwidth, clip=true, trim=5 0 5 5]{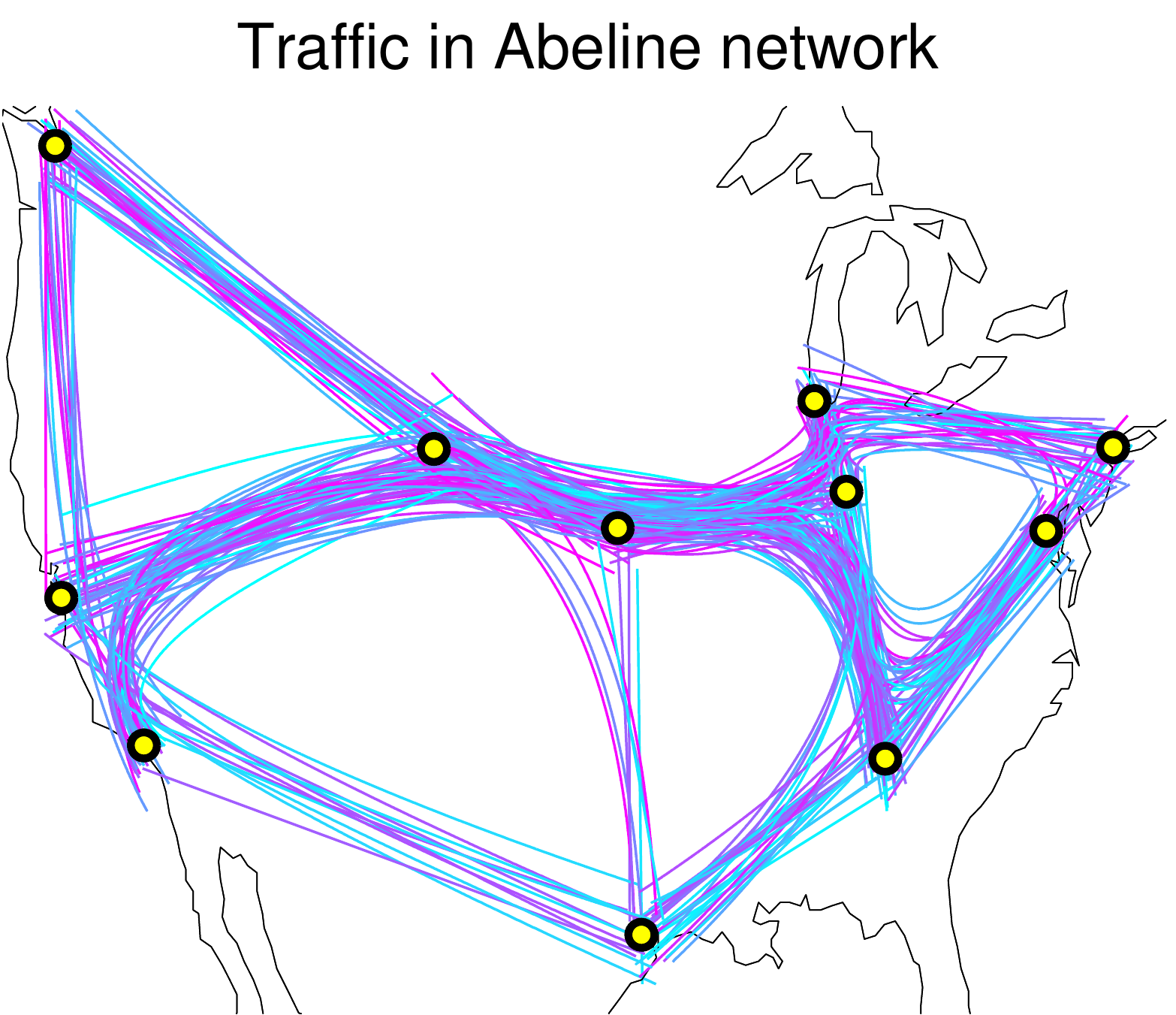}}
\end{minipage}%
\hspace{0.01\textwidth}
\begin{minipage}{.48\linewidth}
\centering
\subfloat[]{\label{}
	\includegraphics[width=\textwidth, height=\textwidth, clip=true, trim=5 0 5 5]{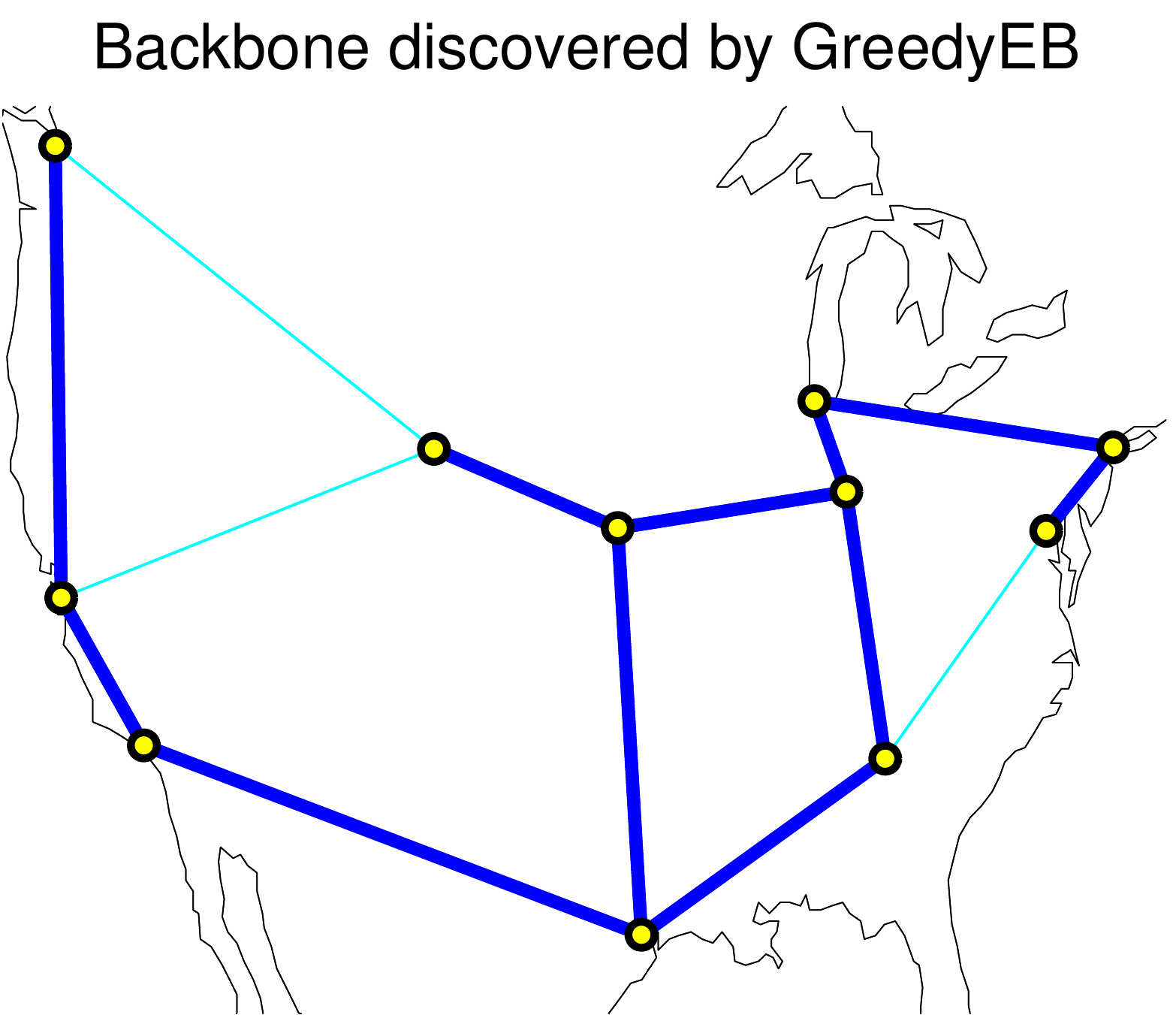}}
\end{minipage}\par\medskip
\caption{Qualitative analysis of the real Internet network.
The figure on the left shows network traffic in the Abeline dataset,
and the one on the right shows the backbone discovered by the
\greedyeb\ algorithm.
As in Figure~\ref{fig:londontube}, the traffic shown is an
interpolation along the shortest path between the source--destination pairs. 
}
\label{fig:abeline}

\end{figure*}

\section{Related work}
\label{sec:related}

As already noted, \backbonediscovery\ is related to the $k$-spanner and the Steiner-forest problem and the
decision versions of both are known to be \NPcomp~\cite{narasimhan2007geometric,williamson2011design}.
The $k$-spanner problem is designed to bound the stretch factor for {\em all pairs} of nodes and not just those from
a specific set of $(s,t)$ pairs. The Steiner-forest problem on the other hand is designed to keep the $(s,t)$ pairs
connected with a minimal number of edges and is agnostic about the stretch factor. 
Both these problems only consider structural information
and completely ignore functional (activity) data that maybe
available about the usage of the network. They also have strict limitations that all nodes need
 to be covered, which makes them restrictive.

The Prize collecting Steiner-forest problem (PCSF)~\cite{hajiaghayi2010prize} is a version of the Steiner-forest problem
 that allows for disconnected source--destination pairs, by imposing a penalty for disconnected pairs. Even in this variant, 
 there is no budget or stretch requirement and hence the optimization problem that PCSF solves is completely different from 
 what we solve. We show how our algorithm fares in comparison to PCSF in Section~\ref{sec:comparison}.

Another enhancement in our work is to normalize edge costs
with measures related to the structure of the network (like edge 
betweenness~\cite{Brandes2007Centrality,girvan2002community-structure,measbetrandom05})
As we show in our experiments, this leads to finding solutions of better quality.


Our work is different from trajectory mining~\cite{giannotti2007trajectory,zheng2009mining}, which consider complete trajectories between source--destination pairs. We do not make use of the trajectories and are only interested in the amount of traffic flowing between a source and destination. Also, the type of questions we try to answer in this paper are different from that of trajectory mining. While trajectory mining tries to answer questions like ``Which are the most used routes between A and B?'', our paper tries to use information about traffic from A to B in order to facilitate a sparse backbone of the underlying network which allows traffic to flow from A to B, also keeping global network characteristics in mind.

The \backbonediscovery\ problem is also related to finding graph
sparsifiers and simplifying graphs. For example, Toivonen et
al.~\cite{ToivonenMZ10} as well as Zhou et al~\cite{ToivonenMZ10b},
propose an approach based on pruning edges while keeping the quality
of best paths between all pairs of nodes, where quality is defined on
concepts such as shortest path or maximum flow.
Misiolek and Chen~\cite{flownetsimpl} propose an algorithm which prune
edges while maintaining the source-to-sink flow for each pair of
nodes.
Mathioudakis et al.~\cite{spine} and 
Bonchi et al.~\cite{spine2}
study the problem of discovering the backbone of a
social network in the context of information propagation, which is a
different type of activity than source--destination pairs, as
considered here.
In the work of Butenko et al.\ a heuristic algorithm for the minimum
connected dominating subset of wireless networks was
proposed~\cite{heurminiwirelessspringer04}.
%
There has been some work in social network research to extract a subgraph
from larger subgraphs subject to constraints~\cite{backdissocnetacm07,netbackdisclusteraxviv12}. 
Other forms of network backbone-discovery have been explored in 
domains including biology, communication networks and the social
sciences. The main focus of most of these approaches 
is on the trade-off between the level of network reduction and the amount of relevant information to be preserved
either for visualization or community detection.
%
While in this paper we try to also sparsify a graph, our objective and approach is completely different from the above because 
we cast the problem in a well-defined optimization framework where the {\em structural} aspects
of the network are captured in the requirement to maintain a low
stretch while the {\em functional} requirements are captured in
maintaining connectedness between traffic terminals, which has not been done before.


In the computer network research community, the notion of software
defined networks (SDN), which in principle decouples the network
control layer from the physical routers and switches, has attracted
a lot of attention~\cite{CasadoFPLGMS09,KimF13}. 
SDN (for example through OpenFlow) will essentially allow network administrators to remotely control routing tables. The \backbonediscovery\ problem
can essentially be considered as an abstraction of the SDN problem, and as we show in Section~\ref{sec:abeline}, our approach can make use of traffic logs 
to help SDN's make decisions on routing and switching in the physical layer.

\section{Conclusions}
\label{sec:conclusions}
We introduced a new problem, \backbonediscovery, to
address a modern phenomenon: these days not only is the {\em structural}
information of a network available but increasingly, highly granular {\em
functional (activity)} information related to network usage is accessible. 
For example, the aggregate traffic usage of the London Subway between all
stations is available from a public website. The \backbonediscovery\
problem allowed us to efficiently combine structural and functional
information to obtain a highly sophisticated understanding of how the
Tube is used (See Figure~\ref{fig:londontube}). 
From a computational perspective, the \backbonediscovery\ problem has elements
of both the $k$-spanner and the Steiner-forest problem and thus requires new 
algorithms to maintain low stretch and connectedness between important
nodes subject to a budget constraint. 
We compare our algorithm with other similar algorithms and show how our algorithm is different and performs better for our setting. Our case studies show the application of the proposed methods for a wide range of applications, including network and traffic planning.

Though our algorithm makes use of shortest paths, in practice, any other types of paths could be incorporated into our algorithm. We leave this generalization for future analysis. 
The use of harmonic mean not only allows us to handle disconnected (s,t)-pairs, but also makes our stretch factor measure more sensitive to outliers. For future work, we would also incorporate a deeper theoretical analysis of the algorithm and the stretch factor measure.

\bibliographystyle{abbrv}
\bibliography{bibliography}
\end{document}